\documentclass[conference]{IEEEtran}
\IEEEoverridecommandlockouts
\usepackage{graphicx}
\usepackage{gensymb}
\usepackage{amsmath}
\usepackage{color}
\usepackage{makecell}
\usepackage{multirow}
\usepackage{siunitx} 
\usepackage{booktabs} 
\usepackage{times} 
\usepackage{color} 
\usepackage{balance} 
\usepackage{url} 
\usepackage{tabularx}
\usepackage{siunitx} 
\usepackage{epstopdf} 
\usepackage{epsfig,tabularx,subfigure,multirow, graphicx}
\usepackage{enumitem}
\usepackage{caption}
\usepackage[bookmarks=false]{hyperref} 

\usepackage{amsthm,amssymb,amsmath}  

\newcolumntype{L}[1]{>{\raggedright\arraybackslash}p{#1}}
\newcolumntype{C}[1]{>{\centering\arraybackslash}p{#1}}
\newcolumntype{R}[1]{>{\raggedleft\arraybackslash}p{#1}}

\usepackage[linesnumbered,ruled,vlined]{algorithm2e}
\SetKwRepeat{Do}{do}{while}
\SetCommentSty{mycommfont}

\long\def\comment#1{}

\newcommand{\nop}[1]{}

\newtheorem{theorem}{\bf Theorem}[section]
\newtheorem{lemma}{\bf Lemma}[section]

\newtheorem{example}{\bf Example}

\theoremstyle{remark}

\theoremstyle{definition}
\newtheorem{definition}{\bf Definition}

\newcommand{\revision}[1]{\color{black}{#1} \color{black}}

\begin{document}

\title{GridTuner: Reinvestigate Grid Size Selection for Spatiotemporal Prediction Models
}

\author{
	{Jiabao Jin{\small $~^{*}$}, Peng Cheng{\small $~^{*}$}, Lei Chen{\small $~^{\dagger}$}, Xuemin Lin{\small $~^{\#, *}$}, Wenjie Zhang{\small $~^{\#}$} }\\
	\fontsize{10}{10}\selectfont\itshape
	$~^{*}$East China Normal University, Shanghai, China\\
	jiabaojin@163.com, pcheng@sei.ecnu.edu.cn\\
	\fontsize{10}{10}\selectfont\itshape
	$~^{\dagger}$The Hong Kong University of Science and Technology, Hong Kong, China\\
	leichen@cse.ust.hk\\
	\fontsize{10}{10}\selectfont\itshape
	$~^{\#}$The University of New South Wales, Australia\\
	\fontsize{9}{9}\selectfont\itshape
	lxue@cse.unsw.edu.au, wenjie.zhang@unsw.edu.au
}

\maketitle

\begin{abstract}
With the development of traffic prediction technology, spatiotemporal prediction models have attracted more and more attention from academia communities and industry.  However, most existing researches focus on reducing model's prediction error but ignore the error caused by the uneven distribution of spatial events within a region. In this paper, we study a region partitioning problem, namely optimal grid size selection problem (OGSS), which aims to minimize the real error of spatiotemporal prediction models by selecting the optimal grid size. In order to solve OGSS, we analyze the upper bound of real error of spatiotemporal prediction models and minimize the real error by minimizing its upper bound. Through in-depth analysis, we find that the upper bound of real error will decrease then increase when the number of model grids increase from 1 to the maximum allowed value. Then, we propose two algorithms, namely Ternary Search and Iterative Method, \revision{to automatically find} the optimal grid size. Finally, \revision{the experiments verify} that the error of prediction has the same trend as its upper bound, and the change trend of the upper bound of real error with respect to the increase of the number of model grids will decrease then increase. Meanwhile, in a case study, by selecting the  optimal grid size, the order dispatching results of a state-of-the-art prediction-based algorithm can be improved up to 13.6\%, which shows the effectiveness of our methods on tuning the region partition for spatiotemporal prediction models.
\end{abstract}

\section{Introduction}

Recently, many spatiotemporal prediction models are proposed  to predict the number of events (e.g., online car-hailing requests \cite{tong2017flexible, wang2020demand, cheng2019queueing}, or street crimes \cite{mohler2011, deepcrime}) within a region (e.g., a gird of 1km$\times$1km) in a period (e.g., next 30 minutes) \cite{zhang2017deep, li2015bike, zhao2016predict}. With the help of predicted information of events, we can improve the platform revenue of online car-hailing systems (e.g., Uber \cite{uber-web}), or reduce crimes effectively through optimizing patrol route of police \cite{deepcrime}.

One common assumption  of spatiotemporal prediction models is that the distribution of spatial events within a region is uniform \cite{tong2017flexible, zhang2017deep, wang2020demand}, which is in fact almost never true in real scenarios. In addition, the selection of region size is mostly decided by experts' experience or simple experimental tests without detailed analysis in many existing research studies:

\begin{itemize}[leftmargin=*]
	\item ``We use 20$\times$30 = 600 grids to cover the cities and one grid represents a 0.01 (longitude)$\times$0.01 (latitude) square'' \cite{tong2017flexible}
	\item The authors use 32$\times$32 grids to cover Beijing area and 16$\times$8 grids for New York City area~\cite{zhang2017deep}. 
	\item ``For the prediction model, DeepST, we set the default grid size as 2km$\times$2km ...'' \cite{wang2020demand}
\end{itemize}

Under the uniform region assumption, the spatiotemporal prediction models are optimized to minimize the model error in model grids (i.e., the difference between the predicted and actual number of spatial events in grids used in the prediction model), which may lead to dramatic real errors in some smaller regions (i.e., the difference between the predicted and the actual number of spatial events in smaller and homogeneous grids). Then, the overall performance of frameworks utilizing spatiotemporal prediction models will not be optimized for real applications. For example, with careful grid size selection, the overall performance of a state-of-the-art prediction based online spatial crowdsourcing framework \cite{tong2017flexible} can be increased up to 13.6\% (shown in a case study in Section \ref{sec:experimental}).

 We illustrate this challenge with the following example:

\begin{figure}\centering
	\scalebox{0.14}[0.14]{\includegraphics{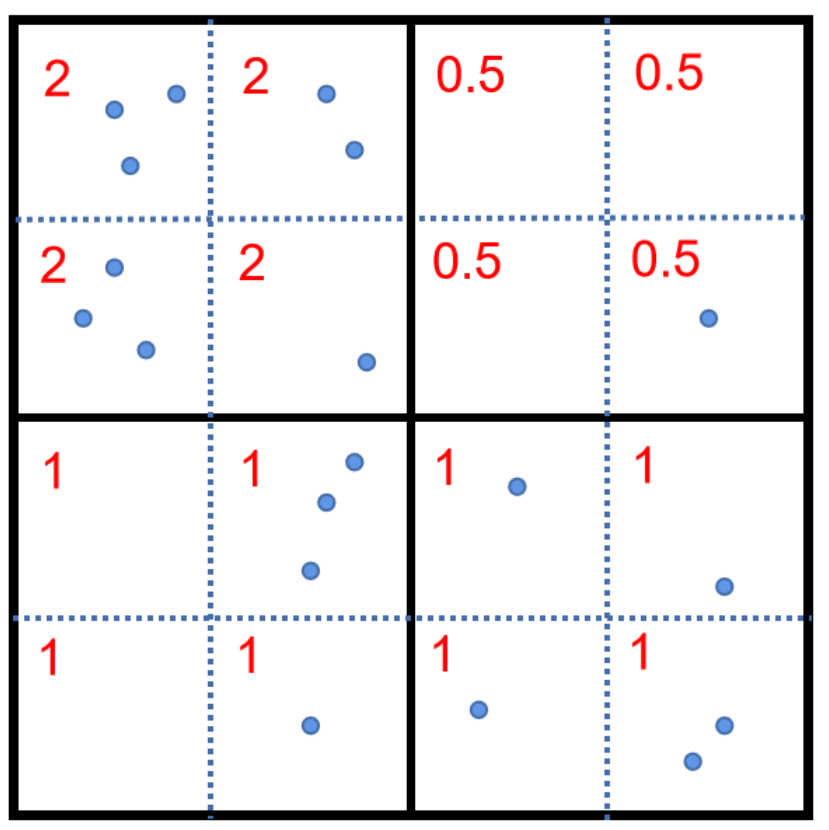}}
	\caption{\small Forecast and Actual Distribution of Orders in  Grids}
	\label{fig:bbq_example}\vspace{-3ex}
\end{figure}

\begin{example}
	\label{exp:expression_error}
	As shown in Figure \ref*{fig:bbq_example}, the solid black lines divide the space into four model grids to be predicted. The blue dot lines further divide each model grid into four smaller grids. We can use spatiotemporal prediction models to predict the number of events in each model grid. In the absence of prior knowledge of the distribution of events within a model grid, the models generally assume that the distribution of events within a model grid is uniform, which means that the number of events in each smaller grid within a same model grid is equal to each other. \revision{Thus, we can estimate the predicted number of each smaller grid through averaging the predicted result of the corresponding model grid. The red number shown in Figure \ref*{fig:bbq_example} denotes the predicted number of events for each smaller grid. The predicted result for each model grid is the summation of the number of events for all its smaller grids.} We can directly calculate the model error of the prediction model on large grids is $3$ (= $|8-9|+|2-1|+|4-4|+|4-5|$). Nevertheless, if the model error is calculated based on smaller grids, it will increase to $10$ \revision{(= $|2-3|+|2-2|+|0.5-0|+|0.5-0|+|2-3|+|2-1|+|0.5-0|+|0.5-1|+|1-0|+|1-3|+|1-1|+|1-1|+|1-0|+|1-1|+|1-1|+|1-2|$).} The reason is that the distribution of events in each large grid is uneven, which is ignored in almost all existing studies. 
\end{example}

\textit{Why not directly predict the spatial events for each smaller grid?} The reason is twofold.
Firstly, due to the uncertainty of spatial events, it is too hard to accurately predict the number of spatial events in a very small grid (e.g., 100m$\times$100m).
\revision{ When the size of grid is too small, there will be no enough historical data for prediction model to learn the distribution of the spatial events in each small area. In addition, the number of spatial events in a small grid is also small, then the accuracy (relative error) of prediction models will be dramatically affected by the randomness of the spatial events, since the uncertainty of spatial events is inevitable \cite{yao2018deep, tong2017flexible, chen2018pcnn}.}
 Secondly, the computation complexity of prediction models will increase remarkably when the number of grids increases \cite{yao2018deep, yu2017spatiotemporal}. Thus, almost all spatiotemporal prediction models still use relatively large grids (e.g., grids of 2km$\times$2km) as the prediction units.

\textit{Can we have an automatic and theoretic-guaranteed optimal grid size selection method to minimize the overall real error of spatiotemporal prediction models?} To overcome this long-standing challenge, in this paper we study the \textit{optimal grid size selection} (OGSS) problem to guide the configuration of grid size such that the real errors are minimized for spatiotemporal prediction models in real applications.

In this paper, we reinvestigate the grid size selection problem in spatiotemporal prediction models in detail. We assume that the distribution of the spatial events in a small enough grid (e.g., 100m$\times$100m) can be considered homogeneous. Then, the real error of a spatiotemporal prediction model is evaluated through the total difference between the predicted number and real number of spatial events among all small grids. Specifically, we decompose the real error of the prediction models into the model error and the expression error. Here, the model error indicates the inherent error of the prediction models, and the expression error stands for the error of using the predicted number of events in a large grid to express the future events in its inner smaller grids (as illustrated in Example \ref{exp:expression_error}). We prove that the summation of model error and expression error of a spatiotemporal prediction model is an upper bound of its real error. We also verify that for any spatiotemporal prediction model, with the increase of the size of grids, the upper bound of its real error will first decrease then increase. Based on our theoretical analysis, we propose two algorithms, namely Ternary Search algorithm and Iterative algorithm, to quickly find the optimal size of model grids for a given spatiotemporal prediction model.

To summarize, we make the following contributions:
\begin{itemize}[leftmargin=*]
	\item We formally define a novel metric, real error, to measure the deviation between the model's forecast and the actual number of events for homogeneous grids. Then, we propose a new problem, namely optimal grid size selection (OGSS), to automatically find the optimal grid size for a given spatiotemporal prediction model in Section \ref{sec:problemDefinition} 
	\item  We analyze the  upper bound of the real error and the relationship between the number of model grids and the  upper bound in Section \ref{sec:erroranalysis}. Then, we propose two algorithms to find the optimal grid size to minimize the upper bound of real error in Section \ref{sec:solution2}.
	\item We conduct  sufficient experiments on different spatiotemporal prediction models to explore the influencing factors of real errors in Section \ref{sec:experimental}.
\end{itemize}

We review the related studies in Section \ref{sec:related} and conclude this paper in Section \ref{sec:conclusion}.

\section{Preliminaries}
\label{sec:problemDefinition}
In this section, we will introduce some basic concepts and present the formal definition of model grid, homogeneous grid, model error,  expression error and real error. We prove that the summation of model error and expression error is an upper bound of real error. 
\subsection{Basic Concepts}
Without loss of generality, in this paper, we consider that the spatiotemporal prediction models will first divide the whole space into $n$ rectangular model grids, then predict the number of spatial events for each model grid in a given future time period. When the size of a grid is small enough (e.g., 100m$\times$100m), the distribution of spatial events can be considered uniform for most application scenarios (e.g., online car-hailing systems). Under this assumption, each model grid will be further divided into $m$ smaller homogeneous grids, where predicted spatial events is considered uniformly distributed in each homogeneous grid. \revision{To ensure HGrids are small enough, it is required that the number of HGrids is larger than $N$ (i.e., $mn>N$). We note $N$ as the minimum number that makes the distribution of spatial events in each HGrid itself is uniform. We propose a method to select a proper value of $N$ in Section \ref{sec:region_depart}.}
We formally define the model grids and the homogeneous grids as follows:
\begin{definition}[Model Grid, MGrid]
	The whole space \revision{is divided} into $n$ same-sized model grids $\{r_1, r_2, \cdots, r_n\}$. The number of actual spatial events happening in $r_i$ in the next period is noted as $\lambda_i$. A spatiotemporal prediction model will predict the number, $\hat{\lambda}_i$, of spatial events  that happening in each model grid $r_i$ in the next period. 
\end{definition}

\begin{definition}[Homogeneous Grid, HGrid]
	Each model grid $r_i$ can be further evenly divided into $m$ homogeneous grids $\{r_{i1}, r_{i2}, \cdots, r_{im}\}$.
	For a homogeneous grid $r_{ij}$, the number of actual spatial events happening in it in the next period is marked as $\lambda_{ij}$ (i.e., $\lambda_i=\sum_{j=1}^{m}\lambda_{ij}$).
\end{definition}

\begin{figure}[t!]\centering\vspace{-3ex}
	\scalebox{0.12}[0.12]{\includegraphics{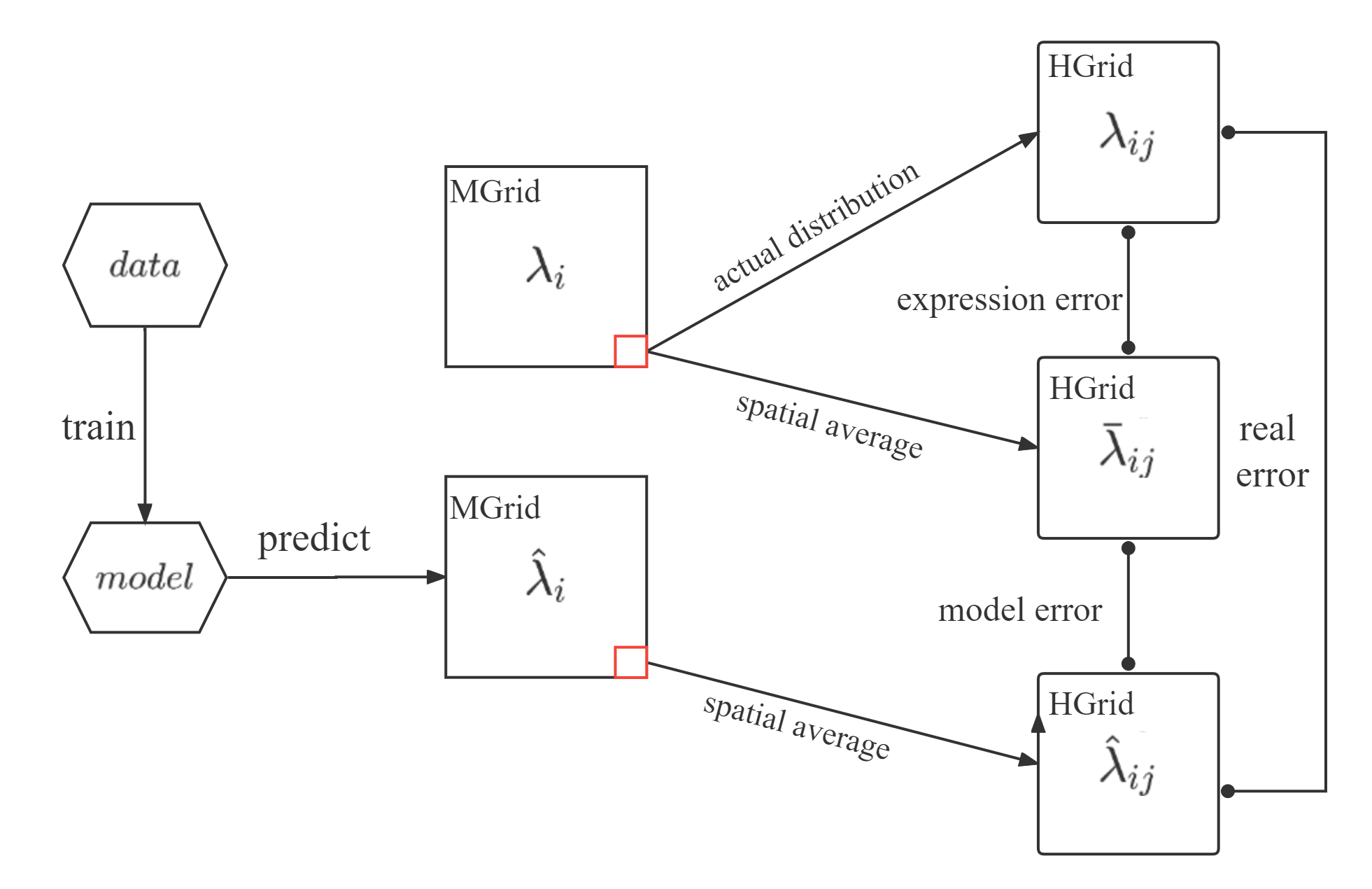}}
	\caption{\small An Illustration of Relationships between Basic Concepts.}
	\label{fig:notation_explain}\vspace{-3ex}
\end{figure}

In the absence of any prior knowledge of the distribution of the spatial events in a model, we assume that the number of spatial events of each HGrid in the MGrid is same to each other according to the principle of maximum entropy \cite{guiasu1985principle}. Thus,  with the actual number, $\lambda_{i}$, of spatial events  in  MGrid $r_i$, the estimated  number of spatial events of  HGrid $r_{ij}$ is denoted as  $\bar{\lambda}_{ij}=\frac{\lambda_i}{m}=\frac{\sum_{j=1}^{m}{\lambda_{ij}}}{m}$.  Similarly, with the predicted number, $\hat{\lambda}_{i}$,  of spatial events of MGrid $r_{i}$, we can have the predicted number of spatial events of HGrid $r_{ij}$ as $\hat{\lambda}_{ij}=\frac{\hat{\lambda}_i}{m}$. 

The differences between $\lambda_{ij}$, $\bar{\lambda}_{ij}$ and $\hat{\lambda}_{ij}$ lead to three types of errors: model error, expression error and real error. Figure \ref{fig:notation_explain} illustrates the three types of errors. The frameworks utilizing spatiotemporal prediction models is hard to get the information about the real distribution of spatial events from the models. Thus, the prediction result $\hat{\lambda}_i$ of a MGrid $r_{i}$ will be divided equally into HGrids without any prior information. The real error describes the difference between the actual number of spatial events $\lambda_{ij}$ and the predicted result $\hat{\lambda}_{ij}$ of HGrid $r_{ij}$. However, it is difficult to calculate the real error directly. Therefore, $\bar{\lambda}_{ij}$ is introduced to decompose the real error into expression error and model error. We formally define the real error model error and expression error as follows:

\begin{definition}[Real Error]
\revision{For a HGrid $r_{ij}$,  its real error $E_r(i,j)$ is defined as the mean/average of difference between its predicted and actual numbers of spatial events in the corresponding same time periods of the historical days for the next period.  It means {\scriptsize$E_r\left(i,j\right)=\mathbb{E}_{\lambda_{ij}\sim P}(|\hat{\lambda}_{ij} - \lambda_{ij}|)$}, where $\lambda_{ij}$ follows a given distribution $P$.}
\end{definition}

In practice, it is difficult to calculate the difference $|\hat{\lambda}_{ij} - \lambda_{ij}|$ without the information about the number $\lambda_{ij}$ of events in next period. Thus, we define the real error as the mean of the difference $|\hat{\lambda}_{ij} - \lambda_{ij}|$. \revision{However, due to the lack of a sufficient number of samples, it is difficult for us to calculate $E_r\left(i,j\right)$ accurately because $\lambda_{ij}$ does not follow the same distribution for different time periods. Another factor is that the environment is prone to change over a long period, so the number of events does not follow the same distribution either. As a result, we use the number of events in the same time period on each day of the previous one month to estimate real error. Suppose that we have a set $\Lambda_{ij}$ of the actual events number, $\lambda_{ij}$, and its corresponding prediction number, $\hat{\lambda}_{ij}$, we can estimate $E_r\left(i,j\right)$ as follows:
{\scriptsize
$$
E_r\left(i,j\right)= \mathbb{E}_{\lambda_{ij}\sim P}(|\hat{\lambda}_{ij} - \lambda_{ij}|)= \frac{1}{\left|\Lambda_{ij}\right|}\sum_{(\hat{\lambda}_{ij},\lambda_{ij})\in \Lambda_{ij}}|\hat{\lambda}_{ij} - \lambda_{ij}|
$$
}
}
\begin{definition}[Model Error]
	\revision{For a HGrid $r_{ij}$, its model error $E_m(i,j)$ is the mean/average of difference between its predicted and estimated numbers of spatial events in the corresponding same time periods of the historical days for the next period. It means {\scriptsize$E_m(i,j)=\mathbb{E}_{\lambda_{ij}\sim P}(|\hat{\lambda}_{ij} - \bar{\lambda}_{ij}|)$}, where $\lambda_{ij}$ follows a given distribution $P$.}
\end{definition}

\begin{definition}[Expression Error]
	\revision{For a HGrid $r_{ij}$, its expression error $E_e(i,j)$ is the mean/average of difference between its estimated and  actual numbers of spatial events in the corresponding same time periods of the historical days for the next period. It means {\scriptsize$E_e(i,j)=\mathbb{E}_{\lambda_{ij}\sim P}(|\bar{\lambda}_{ij} - \lambda_{ij}|)$}, where $\lambda_{ij}$ follows a given distribution $P$.}
\end{definition}

In the previous analysis, we explained that the grid size selection would significantly affect the real error. This paper aims to find an optimal size that minimizes the summation of real errors in all HGrids. We formally define the problem as follows:

\begin{definition}[Optimal Grid Size Selection Problem, OGSS]
	For a given number of all HGrids $N$, and a given model to predict the number of spatial events for MGrids in next period, the optimal grids size selection problem is to find the optimal $n$ to minimize the summation of real error of all HGrids under the constraint $nm>N$, which is:
	{\small\begin{alignat}{2}
    	& \min\limits_{n}&  & \sum_{i=1}^{n}{\sum_{j=1}^{m}{E_r(i,j)}}\\
    	& \text{s.t.}&    \quad & 
    	\begin{aligned}[t]
    	& nm>N
    	\end{aligned}\notag
    \end{alignat}}
\noindent where $m$ represents the minimum required number of HGrids in each MGrid satisfying $nm>N$. 
\end{definition}

\subsection{Upper Bound of Real Error}
\label{sec:ubte}

We denote the summation of model error and expression error as $E_u\left(i,j\right)$ ($=E_m\left(i,j\right)+E_e\left(i,j\right)$). We can prove that $E_u\left(i,j\right)$ is an upper bound on $E_r\left(i,j\right)$ by the theorem \ref{theo:ubte}.

\begin{theorem} [Upper Bound of Real Error]
	$E_u\left(i,j\right)$ is an upper bound of real error $E_r\left(i,j\right)$
	\label{theo:ubte}
\end{theorem}
\begin{proof}
	We prove it through the following inequalities:
{\scriptsize\begin{align*}
		E_r\left(i,j\right)&=\mathbb{E}\left(|\hat{\lambda}_{ij}-\lambda_{ij}|\right) = \mathbb{E}\left(|\hat{\lambda}_{ij}-\bar{\lambda}_{ij}+\bar{\lambda}_{ij}-\lambda_{ij}|\right)\\
		&\leq \mathbb{E}\left(|\hat{\lambda}_{ij}-\bar{\lambda}_{ij}|+|\bar{\lambda}_{ij}-\lambda_{ij}|\right)\\
		&=\mathbb{E}\left(|\hat{\lambda}_{ij}-\bar{\lambda}_{ij}|\right)+\mathbb{E}\left(\left|\bar{\lambda}_{ij}-\lambda_{ij}\right|\right)\\
		&=E_m\left(i,j\right)+E_e\left(i,j\right)=E_u\left(i,j\right)
	\end{align*}}\vspace{-1ex}
\end{proof}

\revision{
Meanwhile, we obtain the upper bound of the difference between $E_u\left(i,j\right)$ and $E_r\left(i,j\right)$ by the following scaling:
{\scriptsize
\begin{align*}
    E_u(i,j)-E_r(i,j)&\leq \mathbb{E}\left(2\min\left(|\hat{\lambda}_{ij}-\bar{\lambda}_{ij}|,|\bar{\lambda}_{ij}-\lambda_{ij}|\right)\right)\\
    &=2\min\left(\mathbb{E}\left(|\bar{\lambda}_{ij}-\lambda_{ij}|\right),\mathbb{E}\left(|\hat{\lambda}_{ij}-\bar{\lambda}_{ij}|\right)\right)\\
    &=2\min\left(E_e(i,j),E_m(i,j)\right)
\end{align*}
}
This indicates that we can ensure that the $E_r\left(i,j\right)$ is small when $E_u\left(i,j\right)$ is minimized. Therefore, we will minimize $E_u\left(i,j\right)$ as much as possible to optimize OGSS in the following sections of this paper. Finally, Table \ref*{table0} shows some important notations used in this paper.
}

\begin{table}
	\centering
	{\small\scriptsize
		\caption{\small Symbols and Descriptions.} \label{table0}
		\begin{tabular}{l|l}
			{\bf Symbol} & {\bf \qquad \qquad \qquad\qquad\qquad Description} \\ \hline \hline
			$r_i$ & a MGrid\\
			$r_{ij}$ & a HGrid in MGrid $r_i$\\
			$n$ & the number of MGrids\\
            $m$ & the number of HGrids for each MGrid\\
            $\bar{\lambda}_{ij}$ & the estimated number of spatial events for HGrid $r_{ij}$\\
            $\lambda_{ij}$ & the actual number of spatial events in HGrid $r_{ij}$\\
            $\lambda_i$ & the actual number of spatial events in MGrid $r_{i}$\\
            $\hat{\lambda}_i$ & the prediction of $\lambda_i$\\
            $\alpha_{ij}$ & \revision{the temporal mean of $\lambda_{ij}$}\\
            $E_r\left(i,j\right)$ & the real error of HGrid $r_{ij}$\\
            $E_e\left(i,j\right)$ & the expression error of HGrid $r_{ij}$\\
            $E_m\left(i,j\right)$ & the model error of HGrid $r_{ij}$\\
			\hline
			\hline
		\end{tabular}
	}\vspace{-2ex}
\end{table}
\section{Error Analysis}
\label{sec:erroranalysis}

In this section, we first explain how to select a proper $N$ such that each HGrid is small enough and can be considered uniform. Then, we discuss the property of expression error and propose two algorithms to quickly calculate expression error. Finally, we analyze the model error.
\revision{
\subsection{Select A Suitable $N$}
}
\label{sec:region_depart}
We explain how to choose a suitable $N$ in this section. Most of spatiotemporal prediction models are based on experience to divide the whole space into many same-sized MGrids (e.g., 2km$\times$2km grid). However, these methods ignore the uneven distribution of spatial events within a MGrid. 

We divide the whole space into {\scriptsize$\sqrt{N}\times \sqrt{N}$} (i.e., {\scriptsize$nm=N$}) same-sized HGrids. \revision{Let $\alpha_{ij}$ be  the mean number of events for  HGrid $r_{ij}$ in the next period, which can be estimated as the average number of the historical records (i.e., nearest one month's data) of $r_{ij}$.} 

\revision{Here, we give the definition of the uniformly distribution for a grid as follows:
\begin{definition}[Uniformly Distributed Grid]
    Given a grid $r_{ij}$ with the expected spatial events number $\alpha_{ij}$ and a  positive integer {\scriptsize$K\in\mathbb{Z}^+$}, we divide the grid into {\scriptsize$K$} smaller grids with the expected spatial events number {\scriptsize$\alpha_{ijk}, k=1,2,...,K$}.  Grid $r_{ij}$ is \textbf{uniformly distributed} if and only if {\scriptsize$\alpha_{ijk}=\frac{\alpha_{ij}}{K}$} for any {\scriptsize$1\leq k\leq K$}. 
\end{definition}
}
Then, we introduce a metric to measure the degree of the uneven distribution for spatial events in HGrids, which is defined as the following formula:\vspace{-1ex}
{\scriptsize\begin{eqnarray}
	D_{\alpha}\left(N\right)=\sum_{i=1}^{n}\sum_{j=1}^{m}\left|\alpha_{ij}-\bar{\alpha}_N \right|,
	\label{eq:d_alpha}
\end{eqnarray}}\vspace{-2ex}

\noindent where  {\scriptsize$\bar{\alpha}_N=\frac{1}{N}\sum_{i=1}^{n}\sum_{j=1}^{m}\alpha_{ij}$}. We notice that when $N$ increases, {\scriptsize$D_{\alpha}\left(N\right)$} will also increase. However, when $N$ is large enough (i.e., spatial events can be considered evenly distributed in each HGrid), {\scriptsize$D_{\alpha}\left(N\right)$} will not significantly increase any more. We prove this with the following theorem:

\begin{theorem}
	Assume that {\scriptsize$N$} is suitable (sufficiently large) such that the spatial events are uniformly  distributed in each HGrid, then {\scriptsize$D_{\alpha}\left(N\right)=D_{\alpha}\left(NK\right)$}, for any {\scriptsize$K\in \mathbb{Z}^{+}$}. 
	\label{the:D_alpha}
\end{theorem}
\begin{proof}
    We divide each HGrid $r_{ij}$ into smaller grids denoted as $r_{ijk}$ ($k=1,2,\dots,K$), where the mean number of events for each smaller grid is denoted as $\alpha_{ijk}$. Due to the uniformity of each HGrid, we have $\alpha_{ijk}=\frac{\alpha_{ij}}{K}$. Thus, we have
    {\scriptsize\begin{eqnarray}
        D_{\alpha}\left(NK\right)&=&\sum_{i=1}^{n}\sum_{j=1}^{m}\sum_{k=1}^{K}\left|\alpha_{ijk}-\bar{\alpha}_{NK}\right| \\ \notag 
        &=&\sum_{i=1}^{n}\sum_{j=1}^{m}K\left|\frac{\alpha_{ij}}{K}-\frac{1}{K}\bar{\alpha}_N\right|
		= D_{\alpha}\left(N\right)  \notag
    \end{eqnarray}
}\vspace{-7ex}
	
\end{proof}
\revision{The increase of $N$ will not contribute to the increase of $D_\alpha\left(N\right)$ when $N$ is large enough, which means that $D_\alpha\left(N\right)$ can be an indicator to help us to select $N$. In other words, we should choose a sufficiently large $N$ so that $D_\alpha\left(N\right)$ is maximized in practice.}

\subsection{Analysis and Calculation of Expression Error}
\label{sec:expression_analyses}
We assume that the number $\lambda_{ij}$ of events in a HGrid $r_{ij}$ follows a poisson distribution $Pois$ with parameter of $\alpha_{ij}$ ($\alpha_{ij}$ is  the mean number of events for  HGrid $r_{ij}$), which is verified in our previous work \cite{cheng2019queueing, cheng2021queueing}. 


\noindent\textbf{Calculation of Expression Error.} We first analyze how to calculate expression error for a given HGrid $r_{ij}$. Due to $\lambda_{ij}\sim Pois(\alpha_{ij})$, we have\vspace{-2ex}
{\scriptsize\begin{equation}
P(\lambda_{ij}=k_h)=e^{-\alpha_{ij}}\frac{{\alpha_{ij}}^{k_h}}{k_h!}, k_h\in\mathbb{N} \label{eq:equation1}
\end{equation}}\vspace{-2ex}

Then, we define the random variable $\lambda_{i,\neq j}$ as the mean of the number of events for all HGrids in MGrid $r_i$ excluding the HGrid $r_{ij}$ (i.e., $\lambda_{i,\neq j}=\sum_{g\neq j}\lambda_{ig}$), and have $\lambda_{i,\neq j}\sim Pois(\sum_{g\neq j}\alpha_{ig})$ because of the additivity of Poisson distribution. Let $\bar{\lambda}_{i,\neq j}=\frac{1}{m}\lambda_{i,\neq j}$, we have\vspace{-2ex}
{\scriptsize\begin{equation}
P(\bar{\lambda}_{i,\neq j}=\frac{k_m}{m})=e^{-\sum_{g\neq j}\alpha_{ig}}\frac{\left({\sum_{g\neq j}\alpha_{ig}}\right)^{k_m}}{k_m!}, k_m\in\mathbb{N} \label{eq:equation2}
\end{equation}}\vspace{-2ex}

Since $\lambda_{ij}$ and $\bar{\lambda}_{i,\neq j}$ is independent of each other, $P(|\bar{\lambda}_{ij}-\lambda_{ij}|)$ can be expressed by
\revision{\vspace{-1ex}
{\scriptsize\begin{eqnarray}
P(|\bar{\lambda}_{ij}-\lambda_{ij}|=\frac{k_d}{m})
&=&P(|\frac{m-1}{m}\lambda_{ij}-\bar{\lambda}_{i,\neq j}|=\frac{k_d}{m}) \notag\\
&=&\sum_{|\frac{m-1}{m}k_h-\frac{k_m}{m}|=\frac{k_d}{m}}P\left(\lambda_{ij}=k_h\right)P\left(\bar{\lambda}_{i,\neq j}=\frac{k_m}{m}\right) \notag\\
&=&\sum_{\frac{(m-1)k_h-k_m}{m}=\pm \frac{k_d}{m}}p\left(r_{ij}, k_h, k_m\right)
\label{eq:equation3}
\end{eqnarray}}
}

\noindent where $p\left(r_{ij}, k_h, k_m\right)=e^{-\sum_{j=1}^{m}\alpha_{ij}}\frac{({\sum_{g\neq j}\alpha_{ig}})^{k_m}(\alpha_{ij})^{k_h}}{k_m!k_h!}$ denoting the probability when the number of events in HGrid $r_{ij}$ is $k_h$ and the number of events in MGrid $r_{i}$ is $k_h+k_m$. Then, we have:
{\scriptsize\begin{eqnarray}
E_e\left(i,j\right)\notag&=&\mathbb{E}(|\lambda_{ij}-\bar{\lambda}_{ij}|)
=\sum_{k_d=0}^{\infty}\frac{k_d}{m}P(|\lambda_{ij}-\bar{\lambda}_{ij}|=\frac{k_d}{m}) \notag\\
&=&\sum_{k_d=0}^{\infty}\frac{k_d}{m}\sum_{\frac{(m-1)k_h-k_m}{m}=\pm \frac{k_d}{m}}p\left(r_{ij}, k_h, k_m\right)\notag\\
&=&\sum_{k_h=0}^{\infty}\sum_{k_m=0}^{\infty} b_{k_hk_m}\label{eq:exp_error}
\end{eqnarray}}

\noindent where $b_{k_hk_m}=\left|\frac{(m-1)k_h-k_m}{m}\right|p\left(r_{ij}, k_h, k_m\right)$. \revision{Here, $p\left(r_{ij}, k_h, k_m\right)$ represents the probability when the number of events in MGrid $r_{i}$ is $k_m+k_h$ and the number of events in HGrid $r_{ij}$ is $k_h$, and the single expression error of HGrid $r_{ij}$ in this situation is $\left|\frac{(m-1)k_h-k_m}{m}\right|$. Thus, Equation \ref{eq:exp_error} can be regarded as a weighted average of the single expression error in all possible cases.}

\revision{We can use Equation \ref{eq:exp_error} to} calculate the expression error of HGrid $r_{ij}$, which also indicates that the expression error is  only related to the $\alpha_{ij}$ of each HGrid $r_{ij}$ and $m$. 

\noindent\textbf{Properties of expression error.} We show that the upper bound of expression error {\scriptsize$E_e\left(i,j\right)$}  is positively correlated to $\alpha_{ij}$ and $m$. In other words, when $\alpha_{ij}$ or $m$ increases, the upper bound of expression error {\scriptsize$E_e\left(i,j\right)$}  will also increase. This relationship is presented with the following lemma:

 \begin{lemma}
	\label{lem:bounded}
	$\forall M_1,M_2\in \mathbb{Z^+}$, we have
	{\scriptsize$$
    \sum_{k_h=0}^{M_1}\sum_{k_m=0}^{M_2}{b_{k_hk_m}}< (1-\frac{2}{m})\alpha_{ij} + \frac{\sum_{k=1}^{m}{\alpha_{ik}}}{m}.
	$$}
\end{lemma}
\begin{proof}\vspace{-2ex}
	{\scriptsize\begin{eqnarray}
		&&\sum_{k_h=0}^{M_1}\sum_{k_m=0}^{M_2}{b_{k_hk_m}}
		=\sum_{k_h=0}^{M_1}\sum_{k_m=0}^{M_2}\left|\frac{(m-1)k_h-k_m}{m}\right|p\left(r_{ij}, k_h, k_m\right) \notag \\
		&\leq&\sum_{k_h=0}^{M_1}\sum_{k_m=0}^{M_2}\left(\frac{(m-1)k_h}{m}+\frac{k_m}{m}\right)p\left(r_{ij}, k_h, k_m\right) \label{eq:lem_proof}
	\end{eqnarray}}
	Considering the \revision{first term of right hand side of Inequation \ref{eq:lem_proof},} we have
	{\scriptsize\begin{eqnarray}
        &&\sum_{k_h=0}^{M_1}\sum_{k_m=0}^{M_2}\frac{(m-1)k_h}{m}p\left(r_{ij}, k_h, k_m\right) \notag \\
        &=&\sum_{k_h=0}^{M_1}\sum_{k_m=0}^{M_2}\frac{(m-1)k_h}{m}e^{-\sum_{j=1}^{m}\alpha_{ij}}\frac{({\sum_{g\neq j}\alpha_{ig}})^{k_m}(\alpha_{ij})^{k_h}}{k_m!k_h!} \notag \\
        &=&\frac{(m-1)}{m}\sum_{k_h=1}^{M_1}{\frac{e^{-\alpha_{ij}}(\alpha_{ij})^{k_h}}{(k_h-1)!}}\sum_{k_m=0}^{M_2}{e^{\sum_{g\neq j}\alpha_{ig}}\frac{({\sum_{g\neq j}\alpha_{ig}})^{k_m}}{k_m!}} \label{eq:berfore_reduce1}\\
        &<&\frac{(m-1)}{m}\sum_{k_h=1}^{M_1}{\frac{e^{-\alpha_{ij}}(\alpha_{ij})^{k_h}}{(k_h-1)!}} \label{eq:2}\\
        &=&\frac{(m-1)\alpha_{ij}}{m}\sum_{k_h=1}^{M_1}{\frac{e^{-\alpha_{ij}}(\alpha_{ij})^{k_h-1}}{(k_h-1)!}} \notag\\
        &=&\frac{(m-1)\alpha_{ij}}{m}\sum_{k_h=0}^{M_1-1}{\frac{e^{-\alpha_{ij}}(\alpha_{ij})^{k_h}}{k_h!}} \label{eq:berfore_reduce2}\\
        &<&\frac{(m-1)}{m}\alpha_{ij} \label{eq:sim_proof}
	\end{eqnarray}}\vspace{-2ex}

	\revision{The item {\scriptsize$e^{-\sum_{g\neq j}\alpha_{ig}}\frac{({\sum_{g\neq j}\alpha_{ig}})^{k_m}}{k_m!}$} in Equation \ref{eq:berfore_reduce1} can be regarded as the probability of {\scriptsize$\tilde{P}\left(x=k_m\right)$} where {\scriptsize$\tilde{P}$} is a Poisson with the mean of {\scriptsize$\sum_{g\neq j}\alpha_{ig}$}. Besides, the item {\scriptsize$\frac{e^{-\alpha_{ij}}(\alpha_{ij})^{k_h}}{k_h!}$} in Equation \ref{eq:berfore_reduce2} can also be regarded as the probability of {\scriptsize$\tilde{P}\left(x=k_h\right)$} where {\scriptsize$\tilde{P}$} is a Poisson with the mean of $\alpha_{ij}$. inequalities \ref{eq:2} and \ref{eq:sim_proof} hold, because as the integral of a Poisson distribution on a part of the whole range is smaller than 1. Then, we can do the same thing to simplify the second term of right hand side of Inequality \ref{eq:lem_proof}::}\vspace{-2ex}

	{\scriptsize\begin{eqnarray}
		&&\sum_{k_h=0}^{M_1}\sum_{k_m=0}^{M_2}\frac{k_m}{m}p\left(r_{ij}, k_h, k_m\right) \notag \\
		&=&\frac{1}{m}\sum_{k_h=0}^{M_1}{\frac{e^{-\alpha_{ij}}(\alpha_{ij})^{k_h}}{k_h!}}\sum_{k_m=1}^{M_2}{e^{-\sum_{g\neq j}\alpha_{ig}}\frac{({\sum_{g\neq j}\alpha_{ig}})^{k_m}}{(k_m-1)!}} \notag\\
		&<&\frac{1}{m}{\sum_{g\neq j}\alpha_{ig}}\notag
	\end{eqnarray}}\vspace{-3ex}

	Thus, we have
{\small$
    \sum_{k_h=0}^{M_1}\sum_{k_m=0}^{M_2}{b_{k_hk_m}}< (1-\frac{2}{m})\alpha_{ij} + \frac{\sum_{k=1}^{m}{\alpha_{ig}}}{m}
	$}\vspace{-1ex}
\end{proof}\vspace{-1ex}

Then, we can get the upper bound of the summation of {\scriptsize$E_e\left(i,j\right)$} for all HGrids is
{\scriptsize$
\sum_{i=1}^{n}{\sum_{j=1}^{m}{E_e\left(i,j\right)}}\leq 2\left(1-\frac{1}{m}\right)\sum_{i=1}^{n}{\sum_{j=1}^{m}{\alpha_{ij}}}
$}.

According to Theorem \ref{theo:ubte}, to minimize the overall real error, we need to minimize the overall expression error. With Lemma \ref{lem:bounded}, to minimize expression error {\scriptsize$E_e\left(i,j\right)$}, we can minimize $\alpha_{ij}$ or $m$.  However, $\alpha_{ij}$ is an inner property of HGrid $r_{ij}$ and not determined by the prediction algorithms. For example, $\alpha_{ij}$ on weekdays and weekends in the same HGrid $r_{ij}$ is quite different. On the other hand, the mean number of events in the same grid can vary greatly in different time periods of the day. Then, \textit{to minimize expression errors, we should select the appropriate $n$ to minimize $m$ under the constraint {\scriptsize$nm>N$}.} Since {\scriptsize$nm>N$}, in order to minimize $m$, we should maximize $n$.

\noindent\textbf{Convergence of Expression Error}. We notice that Equation \ref{eq:exp_error} is a summation of an infinite series. For a HGrid $r_{ij}$, we explain that Equation \ref{eq:exp_error} to calculate expression error $E_e\left(i,j\right)$ can converge:
\begin{lemma}
	Equation \ref{eq:exp_error} to calculate Expression error $E_e\left(i,j\right)$ can converge.
	\label{lem:repr_converge}
\end{lemma}
\begin{proof}
   
	\revision{Considering that {\scriptsize$b_{k_hk_m}$} for any {\scriptsize$k_h,k_m$} is positive and {\scriptsize$\sum_{k_h=0}^{M_1}\sum_{k_m=0}^{M_2}{b_{k_hk_m}}$} is bounded according to the Lemma \ref{lem:bounded}, we can prove that Equation \ref{eq:exp_error} can converge.} Let  
	{\scriptsize$S\left(M_2, M_1\right) = \sum_{k_m=0}^{M_2}\sum_{k_h=0}^{M_1}{b_{k_hk_m}}
	$}.
	Lemma \ref{lem:bounded} shows {\scriptsize$\exists M>0$} make the {\scriptsize$S\left(M_2, M_1\right)\leq M$} hold for any {\scriptsize$M_2 \in \mathbb{Z}$}. Since {\scriptsize$S\left(M_2, M_1\right)$} is monotonically increasing with respect to {\scriptsize$M_2$, $\lim_{M_2\to \infty}S\left(M_2, M_1\right)$} can converge, and we have {\scriptsize$\lim_{M_2\to \infty}S\left(M_2, M_1\right)\leq M
	$}, that is: \vspace{-1ex}
	{\scriptsize$$\sum_{k_m=0}^{\infty}\sum_{k_h=0}^{M_1}{b_{k_hk_m}}\leq M \notag\Leftrightarrow \sum_{k_h=0}^{M_1}\sum_{k_m=0}^{\infty}{b_{k_hk_m}}\leq M$$}\vspace{-1ex}

	Let {\scriptsize$T\left(M_1\right) = \sum_{k_h=0}^{M_1}\sum_{k_m=0}^{\infty}{b_{k_hk_m}}$}, and we have {\scriptsize$\lim_{M_1\to \infty}{T\left(M_1\right)}=E_e\left(i,j\right)$}. In the same way, we can prove that {\scriptsize$E_e\left(i,j\right)$} can converge, which means {\scriptsize$E_e\left(i,j\right)=\lim_{M_1\to \infty} T\left(M_1\right)\leq M
	$} because of the monotone increase with respect to {\scriptsize$M_1$}.
\end{proof}\vspace{-1ex}

Since  Equation \ref{eq:exp_error} can converge, we will introduce algorithms to calculate expression error.

    

\noindent \textbf{Algorithm to Calculate Expression Error}.
Equation \ref{eq:exp_error} shows how to calculate the expression error {\scriptsize$E_e\left(i,j\right)$} for a HGrid $r_{ij}$. In fact, we cannot compute the expression error exactly, but we can prove that the expression error can be approximated to arbitrary precision through the below theorem.

\begin{theorem}
For any $\varepsilon$, there is a number $K$ that makes the following inequality holds:\vspace{-1ex}
{\scriptsize$$
\left|\sum_{k_h=0}^{K}\sum_{k_m=0}^{(m-1)K}b_{k_hk_m}-E_e\left(i,j\right)\right|<\varepsilon
$$}
\label{the:arbitrary_precision}
\end{theorem}\vspace{-2ex}
\begin{proof}
	Lemma \ref{lem:repr_converge} shows that $E_e\left(i,j\right)$ can converge. We have \vspace{-2ex}
	{\scriptsize$$\lim_{\tilde{k}_1\to \infty}\sum_{k_h=0}^{\tilde{k}_1}\sum_{k_m}^{\infty}{b_{k_hk_m}}=E_e\left(i,j\right)$$}\vspace{-1ex}
    
	According to the definition of limit, there will be {\scriptsize$M_1$} for any $\varepsilon>0$ that we have\vspace{-2ex}
	
	{\scriptsize$$\frac{-\varepsilon}{2}<\sum_{k_h=0}^{\tilde{k}_1}\sum_{k_m}^{\infty}{b_{k_hk_m}}-E_e\left(i,j\right)<\frac{\varepsilon}{2}$$}
    when {\scriptsize$\tilde{k}_1>M_1$}, which means\vspace{-2ex}
    
	{\scriptsize\begin{eqnarray}
		\frac{-\varepsilon}{2}+E_e\left(i,j\right)<\sum_{k_h=0}^{\tilde{k}_1}\sum_{k_m}^{\infty}{b_{k_hk_m}}<\frac{\varepsilon}{2}+E_e\left(i,j\right)\notag 
	\end{eqnarray}}\vspace{-2ex}

	Since the series $\sum_{k_h=0}^{\tilde{k}_1}\sum_{k_m}^{\infty}{b_{k_hk_m}}$ are bounded (shown by Lemma \ref{lem:bounded}), we can switch the order of the series.\vspace{-2ex}
	
	{\scriptsize\begin{eqnarray}
		\sum_{k_h=0}^{\tilde{k}_1}\sum_{k_m}^{\infty}{b_{k_hk_m}}=\sum_{k_m}^{\infty}\sum_{k_h=0}^{\tilde{k}_1}{b_{k_hk_m}}\notag 
	\end{eqnarray}}\vspace{-2ex}

	We can also find $M_2$ for a positive number $\varepsilon$ that meets\vspace{-2ex}
	
	{\scriptsize\begin{eqnarray}
		\frac{-\varepsilon}{2}<\sum_{k_h=0}^{\tilde{k}_1}\sum_{k_m}^{\tilde{k}_2}{b_{k_hk_m}}-\sum_{k_m}^{\infty}\sum_{k_h=0}^{\tilde{k}_1}{b_{k_hk_m}}<\frac{\varepsilon}{2}\notag 
	\end{eqnarray}}\vspace{-2ex}

	Based on the definition of the limit when {\scriptsize$\tilde{k}_2>M_2$}, we have \vspace{-2ex}
	{\scriptsize\begin{eqnarray}
		-\varepsilon<\sum_{k_h=0}^{\tilde{k}_1}\sum_{k_m}^{\tilde{k}_2}{b_{k_hk_m}}-E_e\left(i,j\right)<\varepsilon\notag 
	\end{eqnarray}}\vspace{-1ex}
	
	We select a number {\scriptsize$K$} which meets the constraints of {\scriptsize$K>M_1$} and {\scriptsize$(m-1)K>M_2$}. Thus, we have
	{\scriptsize\begin{eqnarray}
        \left|\sum_{k_h=0}^{K}\sum_{k_m=0}^{(m-1)K}b_{k_hk_m}-E_e\left(i,j\right)\right|<\varepsilon
    \end{eqnarray}}\vspace{-2ex}
\end{proof}

Theorem \ref{the:arbitrary_precision} shows that we can achieve the result close to the expression error by selecting a suitable {\scriptsize$K$}. We first need to compute {\scriptsize$p\left(r_{ij}, k_h, k_m\right)$}, which needs {\scriptsize$O\left(k_h+k_m\right)$} time to compute. Then, the complexity of the whole calculation of Equation \ref{eq:exp_error} is  {\scriptsize$O\left(m^2K^3\right)$}. However the calculation of {\scriptsize$p\left(r_{ij}, k_h, k_m\right)$} can be simplified as follows:\vspace{-2ex}

{\scriptsize\begin{eqnarray}
	p\left(r_{ij}, k_h, k_m+1\right)=\frac{\sum_{g\neq j}^{m}{\alpha_{ig}}}{k_m+1}p\left(r_{ij}, k_h, k_m\right) \label{eq:recursive_calculation}\vspace{-2ex}
\end{eqnarray}}\vspace{-1ex}

Based on Equation \ref{eq:recursive_calculation}, Algorithm \ref{algo:repre_algorithm} is proposed to approximately computing the expression error {\scriptsize$E_e\left(i,j\right)$} of the HGrid $r_{ij}$. Since the complexity of computing {\scriptsize$p\left(r_{ij},k_h,k_m\right)$} is {\scriptsize$O\left(1\right)$}, the complexity of Algorithm \ref{algo:repre_algorithm} is {\scriptsize$O\left(mK^2\right)$}.

{\small\begin{algorithm}[t]
	\DontPrintSemicolon
	\KwIn{\small the number $m$ of HGrids per MGrid, $\alpha_{ij}$ for each HGrid $r_{ij}$ in the MGrid $r_i$, \revision{a hyper-parameter $K$}}
	\KwOut{\small the expression error $E_e(i,j)$ of the HGrid $r_{ij}$}
    $E_e(i,j)\gets0$\;
    $\alpha_{i,\neq j}\gets \sum_{g\neq j}^{m}{\alpha_{ig}}$\;
    $p_1\gets e^{-\alpha_{ij}}$\;
    \For{$k_h=1$ to $K$}{
        $p_2\gets e^{-\alpha_{i,\neq j}}$\;
        \For{$k_m=1$ to $(m-1)K$}{
            $\Delta\gets\left|\frac{(m-1)k_h-k_m}{m}\right|p_1p_2$\;
            $E_e(i,j)\gets E_e(i,j)+\Delta$\;
            $p_2\gets \frac{-p_2\alpha_{i,\neq j}}{k_m}$\;
        }
        $p_1\gets \frac{p_1\alpha_{ij}}{k_h}$\;
    }
    \Return $E_e(i,j)$\;
	\caption{\small Expression Error Calculation}
	\label{algo:repre_algorithm}
\end{algorithm}}

\noindent \textbf{Algorithm Optimization}.
Considering the large number of HGrids, even though the time needed to calculate the expression error of each HGrid is only about 0.1 second, the final time cost needed to calculate the summation of expression error of all HGrids with Algorithm \ref{algo:repre_algorithm} is about 4 hours. Therefore, we introduce a more efficient algorithm with time complexity of {\scriptsize$O\left(mK\right)$} in this section based on  a more in-depth analysis of Equation \ref{eq:exp_error}.

According to theorem \ref{the:arbitrary_precision}, we can approximate Equation \ref{eq:exp_error} with the following equations:
{\scriptsize\begin{eqnarray}
	&&\sum_{k_h=0}^{K}\sum_{k_m=0}^{(m-1)K}\left|\frac{(m-1)k_h-k_m}{m}\right|p\left(r_{ij}, k_h, k_m\right) \label{eq:close_to}\\
	&=&\frac{(m-1)}{m}\sum_{k_h=0}^{K}\sum_{k_m=0}^{(m-1)K}k_h\mathbb{I}\left(\left(m-1\right)k_h-k_m\right)p\left(r_{ij}, k_h, k_m\right) \notag\\
	&-&\frac{1}{m}\sum_{k_h=0}^{K}\sum_{k_m=0}^{(m-1)K}k_m\mathbb{I}\left(\left(m-1\right)k_h-k_m\right)p\left(r_{ij}, k_h, k_m\right) \label{eq:equation6}
\end{eqnarray}}
where $\mathbb{I}(x)$ is a indicator function and satisfies:\vspace{-1ex}
{\scriptsize\begin{equation}
	\mathbb{I}(x)=\left\{
	\begin{array}{ll}
	1, & x>0 \\
	-1, & x \leq 0
	\end{array}
	\right. \notag
\end{equation}}\vspace{-1ex}

We transform the \revision{first term} of right hand side of Equation \ref{eq:equation6} (denoted as $e_1$) to the following formula:\vspace{-2ex}

\revision{
{\scriptsize\begin{eqnarray}
	&&\frac{(m-1)}{m}\sum_{k_h=1}^{K}\sum_{k_m=0}^{(m-1)K}k_h\mathbb{I}((m-1)k_h-k_m)p(r_{ij}, k_h, k_m) \notag\\
    &=&\frac{(m-1)}{m}\sum_{k_h=1}^{K}k_h(2{\sum_{k_m=0}^{(m-1)k_h}p(r_{ij}, k_h, k_m)-\sum_{k_m=0}^{(m-1)K}p(r_{ij}, k_h, k_m)})\notag\\
	&=&\frac{(m-1)}{m}\sum_{k_h=1}^{K}{e^{-\sum_{j=1}^{m}{\alpha_{ij}}}\frac{(\alpha_{ij})^{k_h}}{(k_h-1)!}}{e^{'}_1(k_h)}
\end{eqnarray}}
}
Let {\scriptsize$e^{'}_1\left(k_h\right)$} denote a function with respect to {\scriptsize$k_m$}, as follows:
{\scriptsize\begin{eqnarray}
	-\sum_{k_m=0}^{(m-1)K}{\frac{\left(\sum_{g\neq j}{\alpha_{ig}}\right)^{k_m}}{k_m!}}+2\sum_{k_m=0}^{(m-1)k_h}{\frac{\left(\sum_{g\neq j}{\alpha_{ig}}\right)^{k_m}}{k_m!}}\label{eq:equation7}
\end{eqnarray}}
The time complexity of the direct calculation of {\scriptsize$e^{'}_1\left(k_h+1\right)$} is {\scriptsize$O\left(mK\right)$} based on Equation \ref{eq:equation7}\revision{; as a result, the time complexity of calculating $e_1$ is {\scriptsize$O\left(mK^2\right)$}.} \revision{However, we can build the connection between {\scriptsize$e^{'}_1\left(k_h+1\right)$} and {\scriptsize$e^{'}_1\left(k_h\right)$} as follows:
{\scriptsize\begin{eqnarray}
	&&e^{'}_1\left(k_h+1\right)-e^{'}_1\left(k_h\right)\notag \\
    &=&2\sum_{k_m=0}^{(m-1)(k_h+1)}{\frac{\left(\sum_{g\neq j}{\alpha_{ig}}\right)^{k_m}}{k_m!}}-2\sum_{k_m=0}^{(m-1)k_h}{\frac{\left(\sum_{g\neq j}{\alpha_{ig}}\right)^{k_m}}{k_m!}}\notag\\
    &=&2\sum_{k_m=(m-1)k_h}^{(m-1)(k_h+1)}{\frac{\left(\sum_{g\neq j}{\alpha_{ig}}\right)^{k_m}}{k_m!}}\label{eq:equation15}
\end{eqnarray}}\vspace{-1ex}
}

\noindent \revision{Therefore, $e^{'}_1\left(k_h+1\right)$ can be calculated through the result of $e^{'}_1\left(k_h\right)$ so that the time complexity of $e^{'}_1\left(k_h\right)$ can be reduced to $O\left(m\right)$.}
Then we can do the same analysis for the \revision{second term} $e_2$ of right hand side of Equation \ref{eq:equation6}:\vspace{-1ex}
{\scriptsize\begin{eqnarray}
	&&\frac{1}{m}\sum_{k_h=0}^{K}\sum_{k_m=0}^{(m-1)K}k_m\mathbb{I}((m-1)k_h-k_m)p\left(r_{ij}, k_h, k_m\right) \notag\\
	&=&\frac{1}{m}\sum_{k_h=0}^{K}{e^{-\sum_{j=1}^{m}{\alpha_{ij}}}\frac{\left(\alpha_{ij}\right)^{k_h}}{k_h!}}e^{'}_2\left(k_h\right)\notag
\end{eqnarray}}\vspace{-1ex}

\noindent where {\scriptsize$e^{'}_2\left(k_h\right)=-\sum_{k_m=1}^{(m-1)K}{\frac{\left(\sum_{g\neq j}{\alpha_{ig}}\right)^{k_m}}{(k_m-1)!}}+2\sum_{k_m=1}^{(m-1)k_h}{\frac{\left(\sum_{g\neq j}{\alpha_{ig}}\right)^{k_m}}{(k_m-1)!}}$}\revision{, and we can do a similar analysis as Equation \ref{eq:equation15}.}

Algorithm \ref{algo:simple_repre_algorithm} is proposed through the above analysis. By reducing the time complexity of $e^{'}_1\left(k_h\right)$ and $e^{'}_2\left(k_h\right)$ from $O\left(mK\right)$ to $O\left(m\right)$, the time complexity of Algorithm \ref{algo:simple_repre_algorithm} becomes $O\left(mK\right)$.
\begin{algorithm}[t]
	{\small
	\DontPrintSemicolon
	\KwIn{the number $m$ of HGrids per MGrid, $\alpha_{ij}$ for each HGrid $r_{ij}$ in the MGrid $r_i$, \revision{a hyper-parameter $K$}}
	\KwOut{the expression error $E_e(i,j)$ of the HGrid $r_{ij}$}
	
	$p_2\gets 1$;
    $e^{'}_1,e^{'}_2\gets 0$\;
	\For(// initialize $e^{'}_1$ and $e^{'}_2$){$k_m=0$ to $(m-1)K$}{
		$p_2 \gets p_2 \sum_{g\neq j}{\alpha_{ig}}$\;
		$e^{'}_2\gets e^{'}_2-p_2$\;
		$p_2 \gets p_2 / (k_m + 1)$\;
		$e^{'}_1\gets e^{'}_1-p_2$\;
	}
	$p_1\gets e^{-\sum_{j=1}^{m}{\alpha_{ij}}}$\;
    $p_2\gets 1$;
    $e_1,e_2\gets 0$\;
	\For(// calculate the value of $e_1$ and $e_2$){$k_h=1$ to $K$}{
		\For{$k_m=\left(m-1\right)(k_h-1)$ to $\left(m-1\right)k_h$}{
			$e^{'}_2\gets e^{'}_2+2p_2$\;
			$p_2 \gets \frac{p_2}{k_m + 1}$\;
			$e^{'}_1\gets e^{'}_2+2p_2$\;
			$p_2 \gets p_2 \sum_{g\neq j}{\alpha_{ig}}$\;
		}
		$e_1\gets e^{'}_1p_1+e_1$\;
		$p_1 \gets \frac{p_1 \alpha_{ij}}{k_h}$\;
		$e_2\gets e^{'}_2p_1+e_2$\;
	}
	$E_e(i,j)\gets \frac{m-1}{m}e_1-\frac{e_2}{m}$\;
    \Return $E_e(i,j)$\;
	\caption{\small Fast Expression Error Calculation}
	\label{algo:simple_repre_algorithm}}
\end{algorithm}

\subsection{Analysis of Model Error}
\label{sec:prediction_error}
 In this section, we can estimate model error with the mean absolute error for each HGrid. Suppose we use the model $f$ to predict the event number $\hat{\lambda}_i$ (i.e., $\hat{\lambda}_i=f(x_i)$) for the next stage of the MGrid through the historical information of the events {\scriptsize$X$}. Let denote the dataset of each MGrid $r_{i}$ as {\scriptsize$X_{i}$}, and we have {\scriptsize$\cup_{i=1}^{n}X_{i}=X$}. Meanwhile, the number of samples in {\scriptsize$X_{i}$} for each MGrid $r_{i}$ is {\scriptsize$\frac{\left|X\right|}{n}$}. We define the mean absolute error of $f$ as {\scriptsize$MAE(f)$} (i.e., {\scriptsize$MAE(f)=\frac{\sum_{x_i\in \mathbf{X}}{\left|f(x_i)-\lambda_i\right|}}{\left|\mathbf{X}\right|}$}), and we have\vspace{-2ex}
 
{\scriptsize\begin{eqnarray}
	\lim_{\left|\mathbf{X}\right|\to \infty}MAE\left(f\right)&=&\lim_{\left|\mathbf{X}\right|\to \infty}\frac{\sum_{x_i\in \mathbf{X}}{\left|f(x_i)-\lambda_i\right|}}{\left|\mathbf{X}\right|}\notag \\
	&=&\frac{1}{n}\sum_{i=1}^{n}\lim_{\left|\mathbf{X}_i\right|\to \infty}\frac{\sum_{x_j\in \mathbf{X}_i}{\left|f(x_j)-\lambda_j\right|}}{\left|\mathbf{X}_i\right|}\notag\\
	&=&\frac{1}{n}\sum_{i=1}^{n}{E\left(\left|\hat{\lambda}_i-\lambda_{i}\right|\right)}\notag
\end{eqnarray} }\vspace{-1ex}

We can get the relationship between the model error {\scriptsize$E_m\left(i,j\right)$} and {\scriptsize$MAE\left(f\right)$}:\vspace{-1ex}
{\scriptsize\begin{eqnarray}
  \sum_{i=1}^{n}{\sum_{j=1}^{m}{E_m\left(i,j\right)}}&=&\sum_{i=1}^{n}{\sum_{j=1}^{m}{\mathbb{E}\left(\left|\hat{\lambda}_{ij}-\lambda_{ij}\right|\right)}}=\sum_{i=1}^{n}{m\mathbb{E}\left(\left|\hat{\lambda}_{ij}-\lambda_{ij}\right|\right)}\notag\\
  &=&\sum_{i=1}^{n}{\mathbb{E}\left(\left|\hat{\lambda}_i-\lambda_{i}\right|\right)}\approx nMAE(f)\label{eq:prediction_error}
\end{eqnarray}}\vspace{-2ex}

According to Equation \ref{eq:prediction_error}, the total model error will increase when $n$ increases. However, based on the analyses in Section \ref{sec:expression_analyses}, the total expression error will decrease when $n$ increases. We have proved that the summation of expression error and model error is a upper bound of real error, which will first decrease then increase when $n$ increase from 1 to $N$. Thus, to minimize the total real error, we will propose two efficient algorithms to select a proper $n$ in next section.

\section{Search for Optimal Grid Size}
\label{sec:solution2}
From the analysis in Section \ref{sec:erroranalysis},  the size of $n$ will affect expression error and model error of each HGrid $r_{ij}$, which will further affect the upper bound of real error. A straightforward algorithm that checks  all the values of $n$ can  achieve the optimal solution for OGSS with the complexity of {\scriptsize$O(\sqrt{N})$}, which is not efficient. Therefore, we will propose two more efficient algorithms to solve OGSS in this section. We first introduce the upper bound calculation of real error.
\subsection{Calculation of Upper Bound for Real Error }
In practice, it is difficult to directly estimate the real error of each HGrid, and then select the optimal partitioning size.  Theorem \ref{theo:ubte} proves that the summation of expression error and model error is an upper bound of real error. Thus, we can turn to minimize expression error and model error, whose calculations have been discussed in Section \ref{sec:erroranalysis}. Specifically, we can use Algorithm \ref{algo:simple_repre_algorithm} to calculate expression errors and Equation \ref{eq:prediction_error} to estimate model errors.
	
%

Based on the analysis in Section \ref{sec:erroranalysis}, we propose our algorithm showed in Algorithm \ref{algo:opg_1} to calculate $e(\sqrt{n})$ (i.e., {\scriptsize$e(\sqrt{n})=\sum_{i=1}^{n}\sum_{j=1}^{m}E_u\left(i,j\right)$}), which is an approximate problem of OGSS. The time cost of training the model is considerable when calculating the error {\scriptsize$e(\sqrt{n})$}. Therefore, we will introduce two algorithms with fewer computations of {\scriptsize$e(\sqrt{n})$}.

\begin{algorithm}[t]
	{\small
	\DontPrintSemicolon
	\KwIn{\small the number of MGrid $n$, the number of all HGrids $N$, \revision{dataset $\mathbf{X}$,} a prediction method $Model$}
	\KwOut{\small$e(\sqrt{n})$}
    $m\gets$ {\scriptsize$\left\lceil \sqrt{\frac{N}{n}}\right\rceil^2$}\;
    $f\gets Model\left(\mathbf{X}\right)$\;
    $e\gets nMAE\left(f\right)$\;
    divide the global space into $N$ HGrids and estimate the $\alpha_{ij}$ for each HGrid $r_{ij}$\;
    \For{$i=1$ to $n$}{
        \For{$j=1$ to $m$}{
            $e\gets e+ E_e\left(i,j\right)$ \tcp{calculated by Algorithm \ref{algo:simple_repre_algorithm}}
        }
    }
    \Return $e$\;
	\caption{\small $UpperBound\left(n, N, \mathbf{X}, Model\right)$}
	\label{algo:opg_1}}
\end{algorithm}
\subsection{Ternary Search}
Without any prior information, we cannot make any optimization of the most straightforward algorithm. Fortunately, it can be concluded from the analysis in Section \ref{sec:erroranalysis} that the model error will increase and the expression error will decrease when $n$ increases.  It means there exists an equilibrium point that minimizes the summation of the expression error and the model error. Consider an extreme case (i.e., $n=1$), the prediction model only needs to predict the number of events for the whole space in the future, which can be very accurate. For example, according to the historical information of New York City (NYC), the number of spatial events (e.g., rider's order) on weekdays almost keeps a relatively stable value without dramatic fluctuations. At this time, the model error is small, but the expression error is considerable. Even if we could know the exact number of orders in the whole NYC for specific timestamp, it would not help for us to dispatch orders for drivers in a particular street area of NYC. When {\scriptsize$n=N$}, the forecasting model needs to predict a mass of grids' events accurately, which will leads to huge model errors due to the uncertainty of human behavior. While the area of a grid is very small, the uncertainty of human activity will lead huge different of prediction. Therefore, we assume that the trend of {\scriptsize$e(\sqrt{n})$} with the increase of $n$ will first go down and then up (This assumption will be confirmed in Section \ref{sec:real_error_grid_size_exp}).

We propose a ternary search algorithm to find the optimal partition size. Given that $n$ is a perfect square, we need to find the optimal $n$ among {\scriptsize$\sqrt{N}$} numbers. Let $l$ be the minimum of {\scriptsize$\sqrt{n}$} and $r$ be the maximum of {\scriptsize$\sqrt{n}$}. The main idea of ternary search is to take the third-equinox between $r$ and $l$ in each round and then compare the corresponding error of the two third-equinox points denoted as $m_r$, $m_l$. If $e(m_r)>e(m_l)$, let $r=m_r$ for next round; otherwise, let $l=m_l$. The ternary search algorithm showed in Algorithm \ref{algo:opg_2} will drop $\frac{1}{3}$ of possible values for $n$ each time, which results in the convergence.

\begin{algorithm}[t]
	{\small\DontPrintSemicolon
	\KwIn{\revision{dataset $\mathbf{X}$,} prediction model $Model$}
	\KwOut{partition size $n$ that minimize $e(\sqrt{n})$}
    use the method analyzed in Section \ref{sec:region_depart} to select $N$\;
    $l\gets 1$;
    $r\gets \sqrt{N}$\;
    \While{$r-l>1$}{
    $m_r\gets \left\lceil \frac{2}{3}r+\frac{1}{3}l\right\rceil $\;
    $m_l\gets \left\lfloor \frac{1}{3}r+\frac{2}{3}l\right\rfloor $\;
    $e(m_l)\gets UpperBound\left(m_l^2, N, \mathbf{X}, Model\right)$\;
    $e(m_r)\gets UpperBound\left(m_r^2, N, \mathbf{X}, Model\right)$\;
    \eIf{$e(m_l)>e(m_r)$}{
        $l\gets m_l$\;
    }{
        $r\gets m_r$\;
    }
    }
    \eIf{$e(l)>e(r)$}{
    $n\gets r^2$\;
    }{
    $n\gets l^2$\;
    }
    \Return $n$\;
	\caption{\small Ternary Search}
	\label{algo:opg_2}}
\end{algorithm}
If the graph of function {\scriptsize$e(\sqrt{n})$} has only one minimum point, then the ternary search will find the optimal solution. However,  the graph of function {\scriptsize$e(\sqrt{n})$} is not always ideal, but the ternary search algorithm can also find a good solution.

\noindent \textbf{Time Complexity.} For a given $N$, we can mark the algorithm complexity as {\scriptsize$T(\sqrt{N})$}. We know that the algorithm will drop $\frac{1}{3}$ of the possible values from the above analysis, thus converting the original problem into a subproblem. Therefore, we have: {\scriptsize$T(\sqrt{N})=T(\frac{2}{3}\sqrt{N})+2$}.
We can infer that the time complexity of Algorithm \ref{algo:opg_2} is {\scriptsize$O(\log{\sqrt{N}})$} according to the master theorem.

\subsection{Iterative Method}
\label{sec:iterative}
Although the ternary search algorithm reduces the algorithm complexity from {\scriptsize$O(\sqrt{N})$} to {\scriptsize$O(\log{\sqrt{N}})$} based on the traversal algorithm, the experiments in Section \ref{sec:experimental} show that the ternary search algorithm may miss the optimal global solution in some situations. Therefore, we will introduce an iteration-based algorithm with a greater probability of achieving the optimal $n$ in this section.


Considering that the upper bound on the real error is large when   $n$ is either large or  small, the global optimal  value for $n$ tends to be somewhere in the middle. We can roughly choose the possible value $p$ of the optimal solution through practical experience and then take this value as the initial position to conduct a local search. We set a search boundary $b$, and if the size of error for the current position is smaller than all possible regions within the boundary $b$, the current position is likely to be the optimal solution. In order to speed up the search process, we start the current position of searching from the boundary $b$ to avoid local traversal when {\scriptsize$e(\sqrt{n})$} is monotonous.  The details of the algorithm is shown in Algorithm \ref{algo:opg_3}.

In Algorithm \ref{algo:opg_3}, the choice of the initial position $p$ and the setting of the search boundary $b$ significantly affect the quality of its result and its efficiency. 
Based on the experience from the existing studies \cite{wang2020demand}, we use the default grid of $2km\times 2km$ (i.e., approximately $16\times 16$) as the corresponding initial position to speed up the search for the global optimal $n$. 
On the other hand, the search boundary $b$ has an essential influence on the properties of the solution and the algorithm's efficiency. When the search boundary is large, the probability of the algorithm finding the optimal solution will increase, but the efficiency of the algorithm execution will decrease. On the contrary, the algorithm can converge quickly when the search boundary is small with a small probability of finding the optimal solution.

\begin{algorithm}[t]
	{\small\DontPrintSemicolon
		\KwIn{\revision{dataset $\mathbf{X}$,} prediction model $Model$}
		\KwOut{partition size $n$ that minimize $e(\sqrt{n})$}
		use the method analyzed in Section \ref{sec:region_depart} to select $N$\;
		$p\gets 16$; $b\gets 4$\;
		$flag\gets true$\;
		\While{$flag$}{
			$flag\gets false$\;
			\For{$i=b$ to $1$}{
                $e(p+i)\gets UpperBound\left((p+i)^2, N, \mathbf{X}, Model\right)$\;
                $e(p-i)\gets UpperBound\left((p-i)^2, N, \mathbf{X}, Model\right)$\;
				\If{$e(p) > e(p + i)$}{
					$p \gets p + i$\;
					$flag \gets true$\;
					break\;
				}
				\If{$e(p) < e(p - i)$}{
					$p \gets p - i$\;
					$flag \gets true$\;
					break\;
				}
			}
		}
        $n\gets p^2$\;
		\Return $n$\;
		\caption{\small Iterative Method}
		\label{algo:opg_3}}
\end{algorithm}
\section{Experimental Study}
\label{sec:experimental}

\subsection{Data Set}

We use realistic data to study the property of expression error and model error. 

\revision{\textbf{New York Taxi Trip Dataset.} New York Taxi and Limousine Commission (TLC) Taxi Trip Data \cite{nyc-web} includes the taxi orders in NYC. We use the Taxi Trip Dataset from January to May 2013 (i.e., January to April as  training  set, May 1st to 27th as  validation  set, and May 28th as  test set). There are 282,255 orders in test set. The size of the whole space is $23km\times 37km$ (i.e., \ang{-73.77}$\sim$\ang{-74.03}, \ang{40.58}$\sim$\ang{40.92}).} Since the number of other types of taxis in NYC is much smaller than that of yellow taxi, we only use the trip data of yellow taxi. Each order record contains the pick-up and drop-up locations, the pick-up timestamp, and the driver's profit.

\revision{\textbf{Chengdu Taxi Trip Dataset.} DiDi Chuxing GAIA Open Dataset \cite{gaia} provides taxi trips in Chengdu, China. We use the taxi trip records from November 1st, 2016 to November 25th, 2016 as training set, November 26th to 29th, 2016 as validation set and November 30th, 2016 as test data set. There are 238,868 orders in test set. The size of Chengdu is also $23km\times 37km$ (i.e., \ang{103.93}$\sim$\ang{104.19}, \ang{30.50}$\sim$\ang{30.84}).}

\revision{\textbf{Xi'an Taxi Trip Dataset.} DiDi Chuxing GAIA Open Dataset \cite{gaia} also provides a dataset of taxi trips in Xi'an, China. We use the taxi trip records from October 1st, 2016 to October 25th, 2016 as training data set, October 26th to 29th, 2016 as validation set and October 30th, 2016 as test set. There are 109,753 orders in test set. The size of Xi'an is $8.5km\times 8.6km$ (i.e., \ang{108.91}$\sim$\ang{109.00}, \ang{34.20}$\sim$\ang{34.28}).}

\revision{Please refer to Appendix A for the distributions of the datasets.}

\subsection{Experiment Configuration}
\revision{We use three prediction models to predict the numbers of  future spatial events in different regions:}

\revision{\textbf{Multilayer Perceptron (MLP) \cite{rosenblatt1961principles}}: We use a neural network consisting of six fully connected layers. The numbers of hidden units on each layer are 1024, 1024, 512, 512, 256 and 256. When the size of MGrid is $n$, we can get the model input (8, $\sqrt{n}$, $\sqrt{n}$), which represents the number of all regions in nearest eight time slots, and we use a flatten layer to map the original input to a vector with the size of $8\times n$ before it is fed into the model.}

\textbf{DeepST} \cite{zhang2017deep}: DeepST divides a day into 48 time slots (i.e., 30 minutes per time slot) and calculates inflow and outflow of the events. As a result, DeepST can calculate the number of events in the next time slot by predicting the inflow and outflow status of events in the next time slot. It uses three types of historical information: closeness, period and trend. Closeness expresses the number of events in the nearest eight time slots, period expresses the number of events at the same time slot of the previous eight days, and trend represents the number of events at the same time slot of the previous eight weeks. \revision{DeepST mainly utilizes the spatial information to predict the spatial events for next time slot.}

\revision{\textbf{Dmvst-Net} \cite{yao2018deep}: Dmvst-Net models the correlations between future demand and recent historical data via long short term memory (LSTM) and models the local spatial correlation via convolutional neural network (CNN). Moreover, Dmvst-Net models the correlations among regions sharing similar temporal patterns. Compared with DeepST, Dmvst-Net utilizes both spatial and temporal information, which leads to a better performance of the prediction model.}

Since the size of model input for DeepST and Dmvst-Net is different in the experiment, we need to map the original input to the same $shape$  to ensure that the model structure does not change significantly through a conditional deconvolution layer. When the number of MGrid is $n$, that is, the input dimension of the model is $(2, \sqrt{n}, \sqrt{n})$, the size $k$ of the convolution kernel and step size $s$ of the convolutional layer can be obtained through the following formula:
{\small\begin{eqnarray}
	s&=&\left\lfloor \frac{shape}{n-1}\right\rfloor \notag\\
	k&=&shape-s\left(n-1\right) \notag
\end{eqnarray}}
Here, we set $shape=128$ in our experiment. Finally, we add a convolution layer with the same stride and the same size as the deconvolution layer as the last layer of DeepST.

\begin{table}\vspace{-2ex}
	\centering
	{\small\scriptsize
		\caption{\small Experiment Setting for Training Model} \label{tab:experiment_configure}
		\revision{
			\begin{tabular}{c|c}
				{\bf Symbol} & {\bf Setting} \\ \hline \hline
				$N$ & {$128\times128$}\\
				$n$ & {4$\times$4,$\dots$,\textbf{16$\times$16},$\dots$,75$\times$75,76$\times$76}\\
				time slot & {30 minutes}\\ 
				prediction model & {MLP, \textbf{DeepST}, Dmvst-Net}\\
				\hline
				\hline
			\end{tabular}
		}
	}\vspace{-2ex}
\end{table}

\revision{As the dataset used in this experiment is the Taxi Trip Dataset, \textbf{Order Count Bias} is used as the metric of model error, expression error and real error in this experiment. Model error represents the difference between the predicted order quantity and the estimated order quantity; expression error represents the difference between the estimated order quantity and the actual order quantity; real error represents the difference between the actual order quantity and the predicted order quantity. Considering that we will constantly change the grid size in the experiment, it is meaningless to consider the error of a single grid; therefore, the errors we discuss in subsequent experiments are the summation of errors of all grids, unless otherwise specified.}

\begin{figure*}[t!]\vspace{-3ex}
	\begin{minipage}{0.245\linewidth}\centering
		\scalebox{0.4}[0.4]{\includegraphics{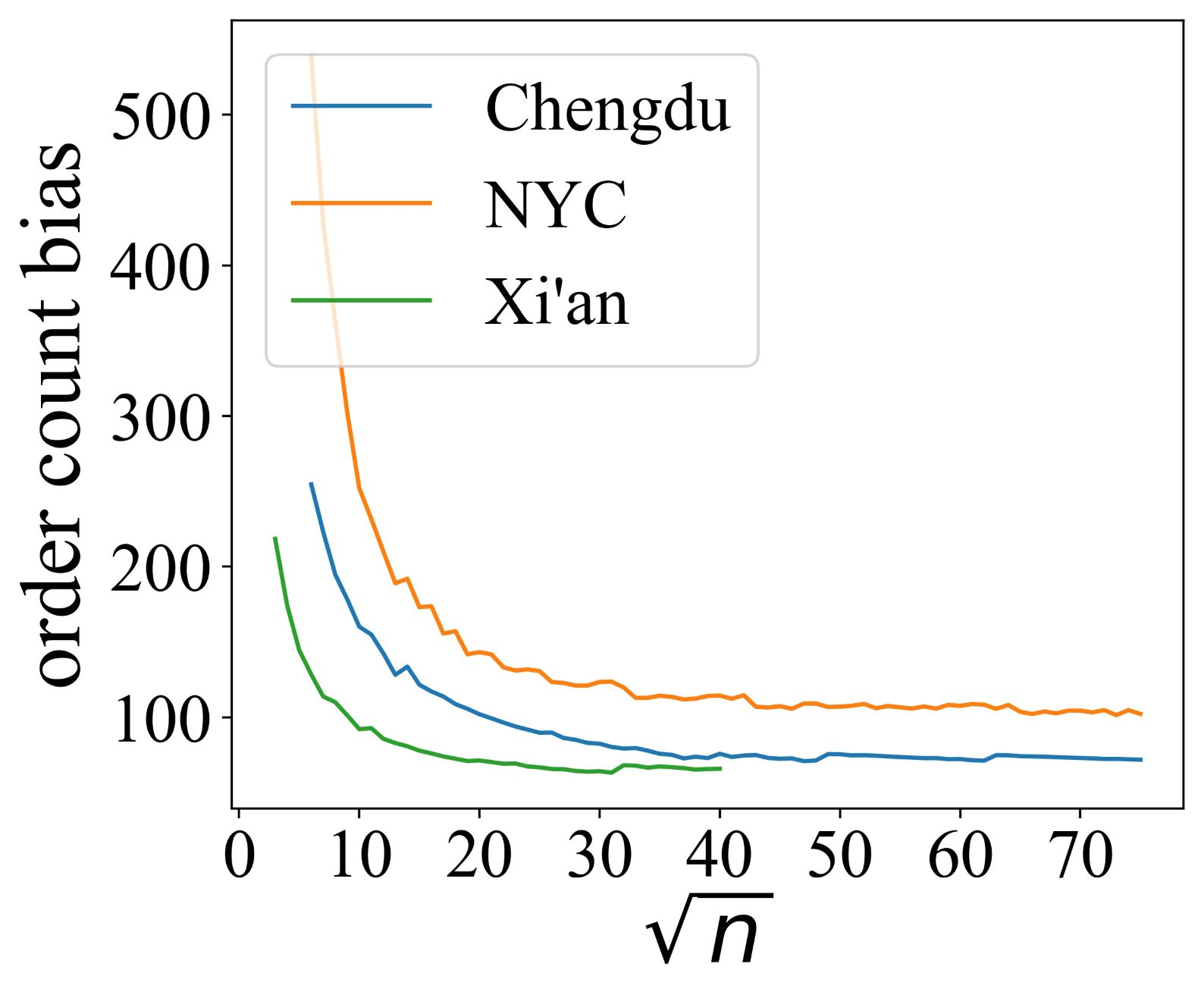}}
		\caption{\small Effect of $n$ on Expression Error in Different Cities}
		\label{fig:multi_citys_expression_error}\vspace{-2ex}
	\end{minipage}
	\begin{minipage}{0.745\linewidth}\centering
		\subfigure[][{\scriptsize Chengdu}]{
			\scalebox{0.4}[0.4]{\includegraphics{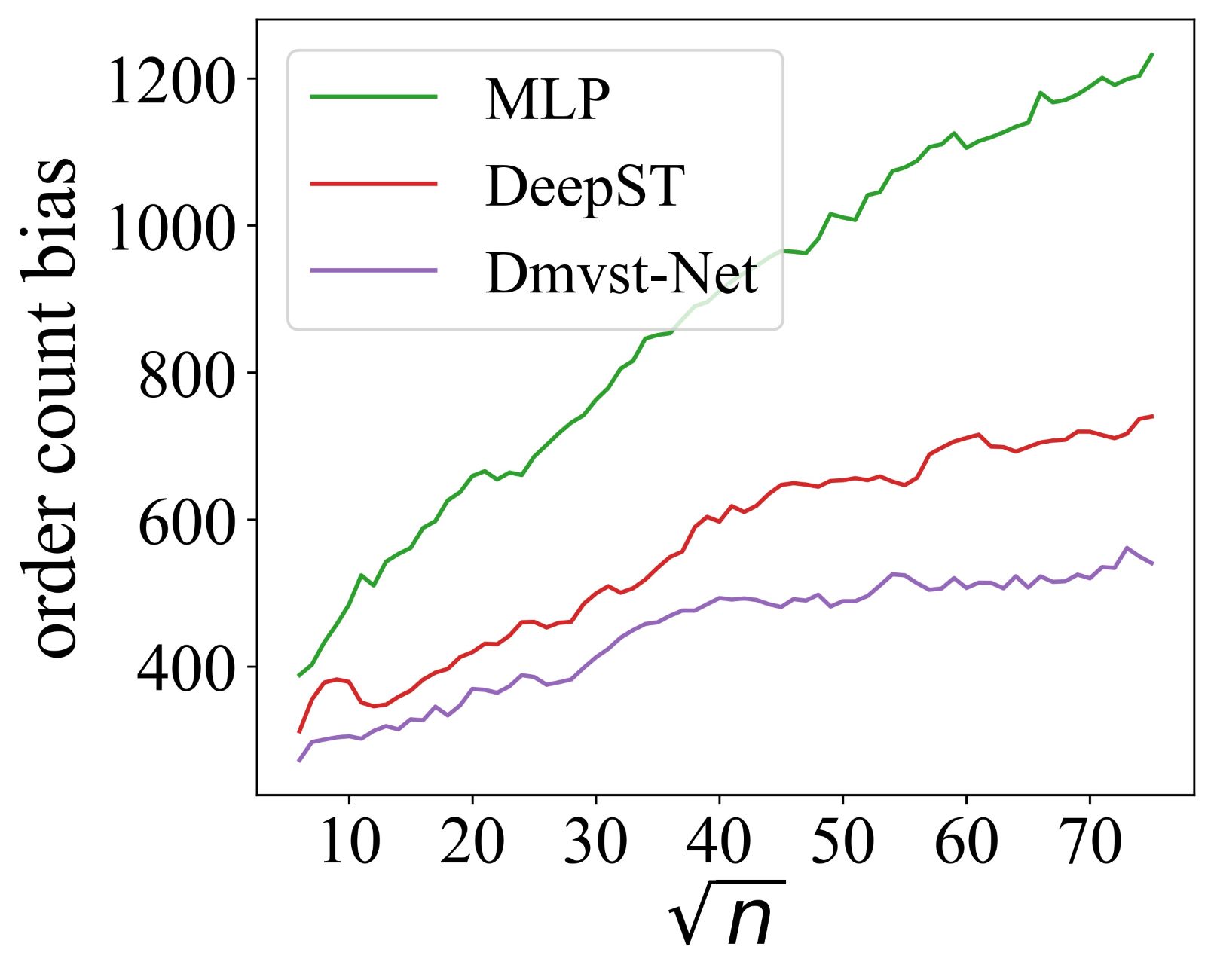}}
			\label{fig:chengdu_deepst_model_error}}
		\subfigure[][{\scriptsize NYC}]{
			\scalebox{0.4}[0.4]{\includegraphics{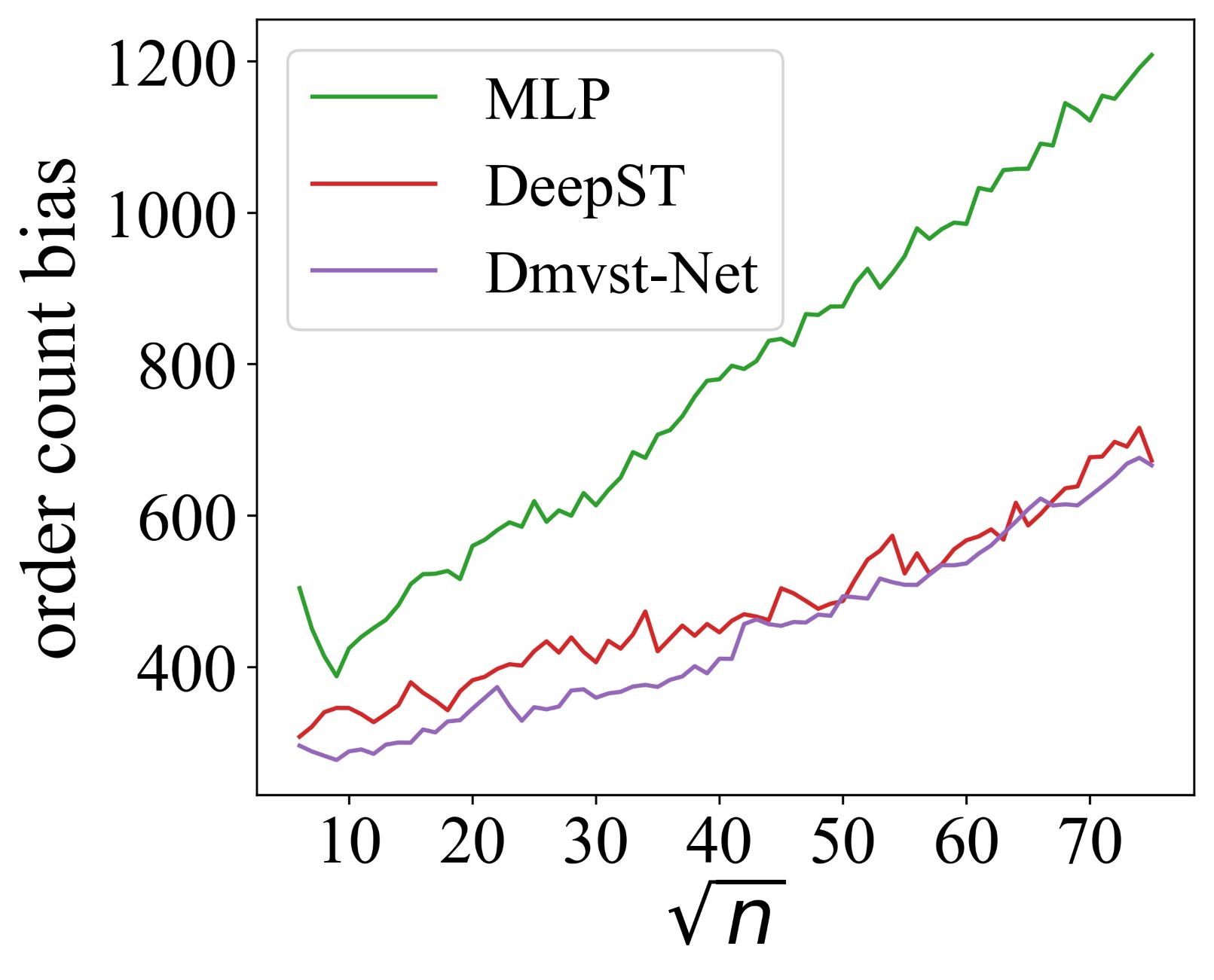}}
			\label{fig:nyc_deepst_model_error}}
		\subfigure[][{\scriptsize Xi'an}]{
			\scalebox{0.4}[0.4]{\includegraphics{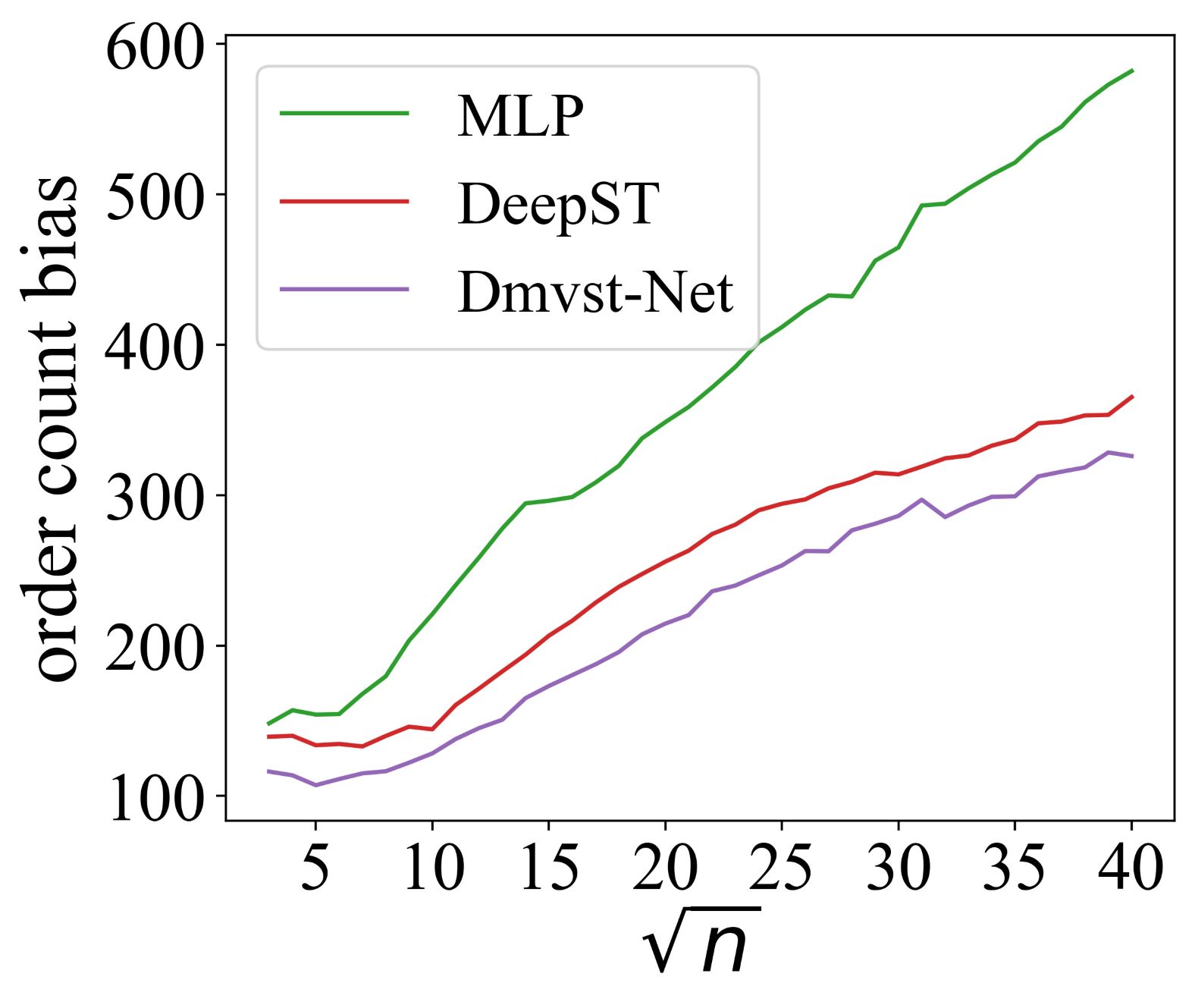}}
			\label{fig:xian_deepst_model_error}}
		\caption{\small Effect of $n$ on the Model Error}
		\label{fig:model_error}\vspace{-2ex}
	\end{minipage}
\end{figure*}

\begin{figure*}[t!]
	\centering
	\subfigure[][{\scriptsize Real Error in Xi'an}]{
		\scalebox{0.48}[0.48]{\includegraphics{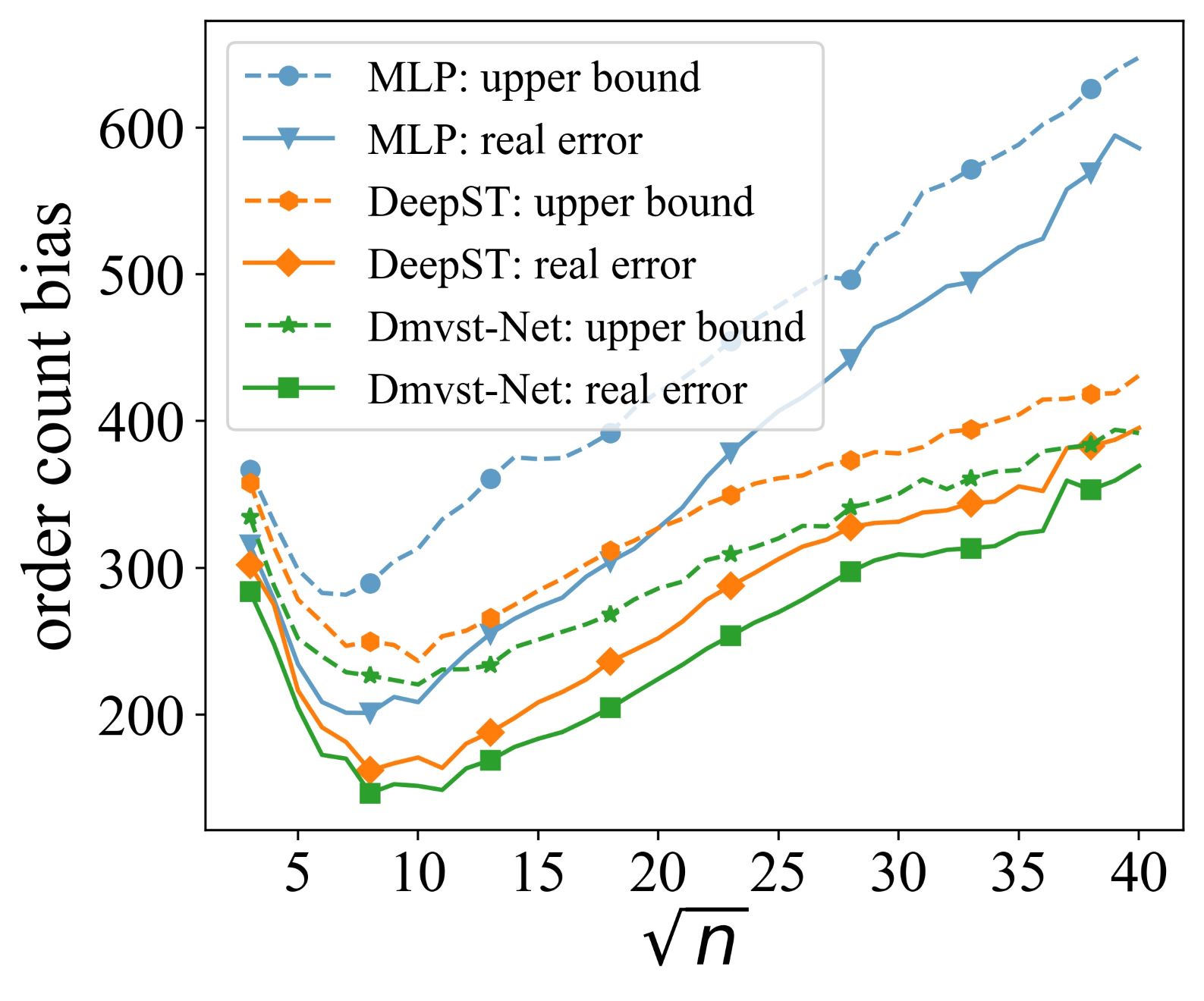}}
		\label{fig:xian_real_error}}
	\subfigure[][{\scriptsize Real Error in Chengdu}]{
		\scalebox{0.48}[0.48]{\includegraphics{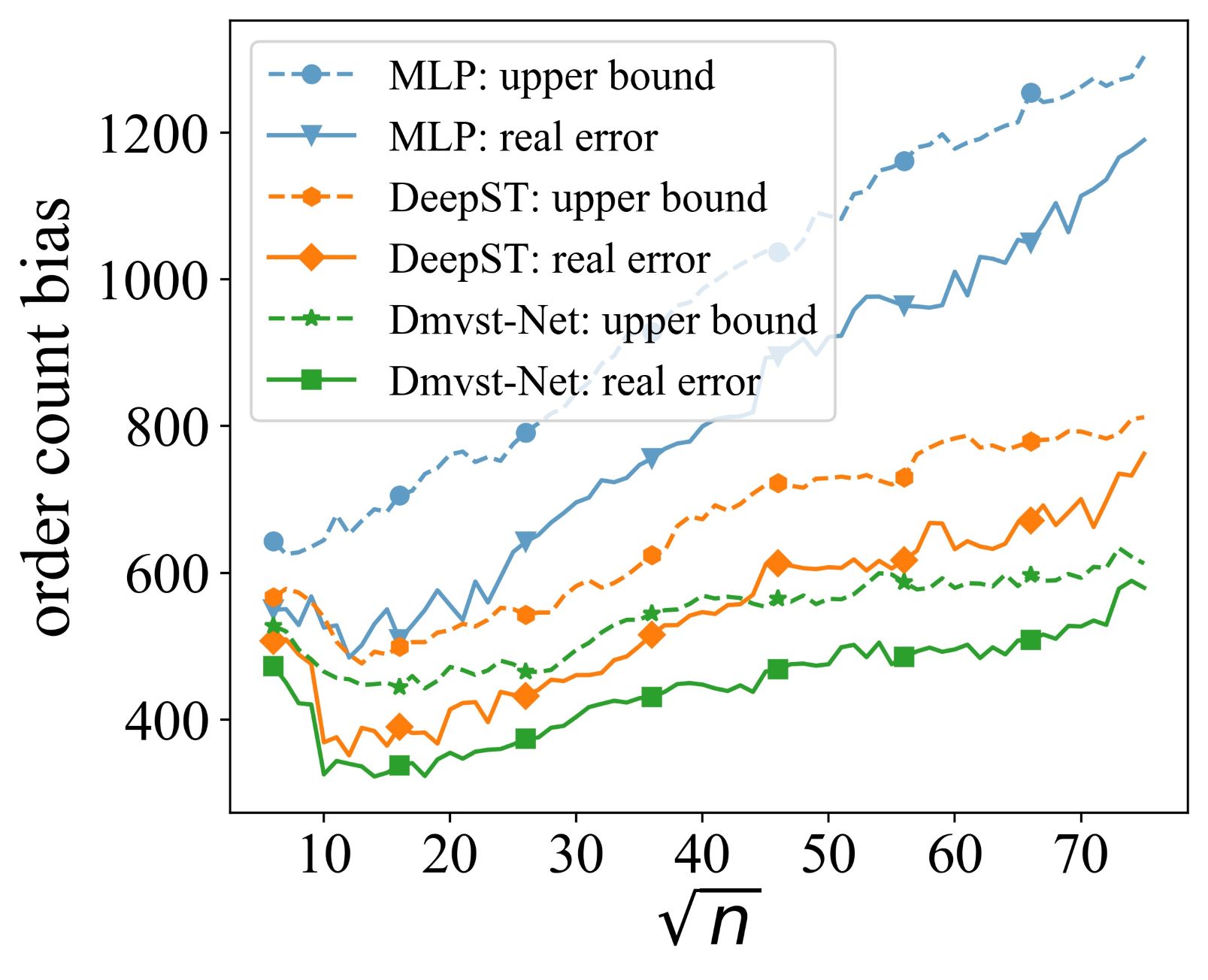}}
		\label{fig:chengdu_real_error}}\vspace{1ex}
	\subfigure[][{\scriptsize Real Error in NYC}]{
		\scalebox{0.48}[0.48]{\includegraphics{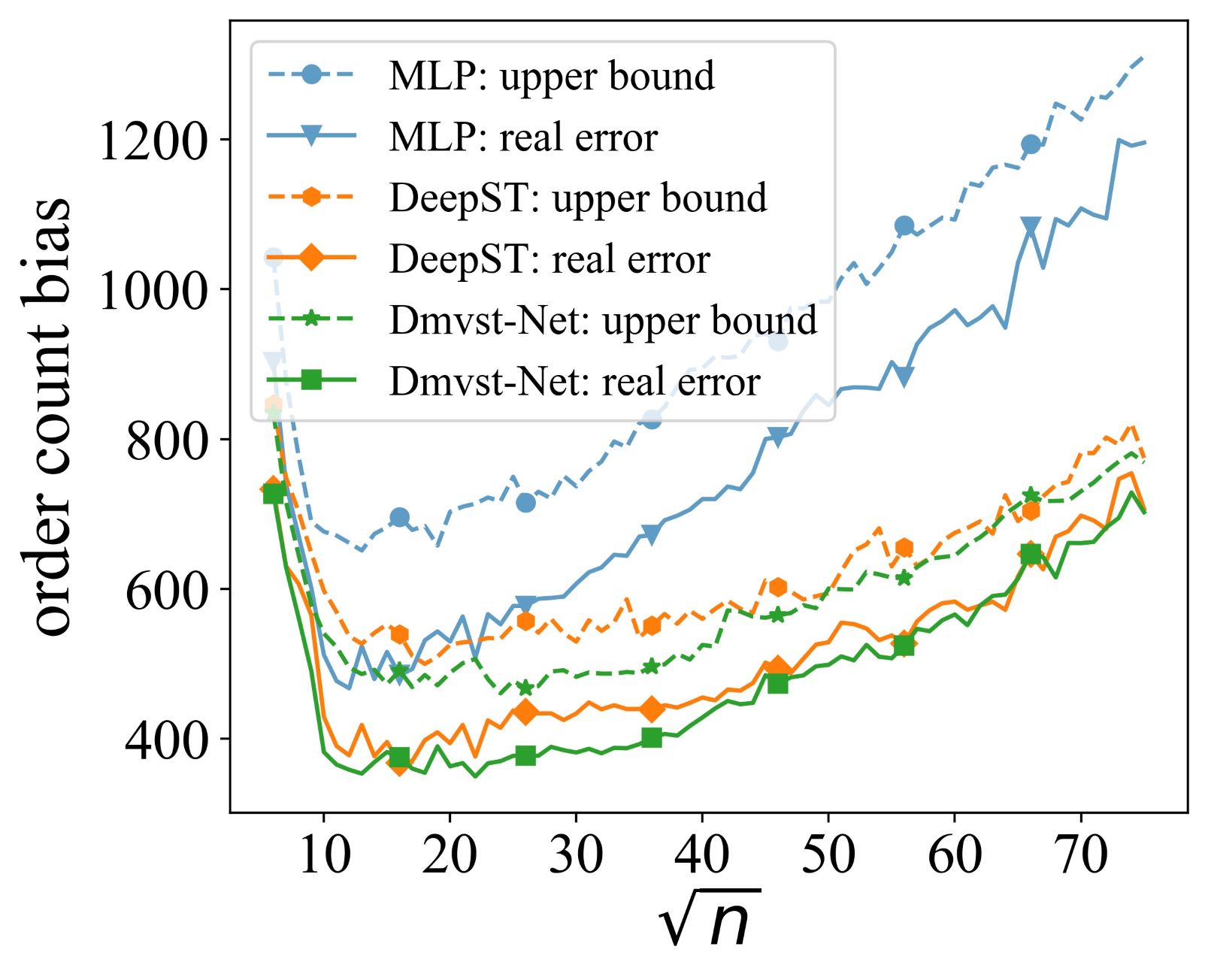}}
		\label{fig:nyc_real_error}}\vspace{-1ex}
	\caption{\small Effect of $n$ on Real Error in Different Cities with Different Prediction Models}
	\label{fig:citys_models_real_error}\vspace{-2ex}
\end{figure*}

In order to calculate the expression error of a HGrid, we need to estimate the mean number of events $\alpha_{ij}$ for the grid $r_{ij}$ in advance. Over a long period, grid environments will change significantly so that the number of events for the same grid does not follow the same distribution. On the other hand, when sample size is small, the estimate of the mean number for events will produce a considerable bias. Therefore, when estimating the mean number of events, we need to choose the appropriate range of adoption. At the same time, considering the remarkable difference about the number of events at different periods in a day and the great difference in the willingness of people to travel on weekdays and workdays, this experiment takes the average number of events at the same period of all workdays in last one month as the mean number $\alpha_{ij}$ of events in the HGrid $r_{ij}$. In subsequent experiments, we estimate $\alpha_{ij}$ by using the number of events between $8:00$ $A.M.$ and $8:30$ $A.M.$ as default unless otherwise stated. 
\revision{
The above experimental settings are summarized in Table \ref{tab:experiment_configure} where the default parameters are in bold. Our experiments are run on AMD Ryzen 5-5600H with 32 GB RAM and GeForce RTX 3050 in Python, while LS, POLAR and DAIF in Java.
}

\revision{
\subsection{Relationship between Real Error and $n$}
\label{sec:real_error_grid_size_exp}
In this section, we mainly show the effect of $n$ on  the expression error and the model error as analyzed in Section \ref{sec:erroranalysis} and verify that real error has the same change trend as its upper bound. \\
\textbf{Expression Error.} We use Algorithm \ref{algo:simple_repre_algorithm} to calculate the expression errors in different cities, which all decrease with the increase of $n$ as shown in Figure \ref{fig:multi_citys_expression_error}. Since orders in NYC are more evenly distributed than in Chengdu, therefore, the expression error of Chengdu is smaller than that of NYC when $n$ is the same. Additionally, the order quantity of Xi'an is much smaller than that of the other two cities. Meanwhile, the order distribution of Xi'an is more uniformly distributed compared with the other two cities. As a result, the expression error of Xi'an is much smaller than that of other cities. We analyze the relationship between expression error and the uniformity of order distribution in detail in Appendix B.\\
\textbf{Model Error.} We test the performance of three prediction models (i.e., MLP, DeepST and Dmvst-Net) on the datasets of NYC and Chengdu as shown in Figures \ref{fig:model_error}. The experimental results show that the model error of the three prediction models all increase with the increase of $n$ on the two data sets. The model errors of DeepST and Dmvst-Net are much smaller than that of MLP with relatively simple model structure, while Dmvst-Net makes use of time information of historical data so that it performs better than DeepST.\\
\textbf{Real Error.} Figure \ref{fig:citys_models_real_error} shows the relationship between real error and its upper bound in different cities while using different prediction models. The real error and its upper bound have the same trend, all falling first and then rising while changing $n$. Comparing with Chengdu, the expression error of NYC is larger, which makes the optimal $n$ of NYC larger than that of Chengdu when the same prediction model is used. For example, the real error of NYC based on Dmvst-Net is also small when $n$ is $30\times30$ as shown in Figure \ref{fig:nyc_real_error}. On the other hand, the prediction model with higher accuracy makes the real error significantly smaller, and also leads to the increase of $n$ that minimizes the real error. Taking NYC as an example, the optimal value of $n$ is 23 when using Dmvst-Net as prediction model; when the prediction model is DeepST, the optimal value of $n$ is 16; when the prediction model is MLP, the optimal value of $n$ is 13. In the case of models with high accuracy, a larger $n$ helps to reduce expression error. Moreover, when we use MLP as a prediction model to forecast the number of orders in Chengdu, Figure \ref{fig:chengdu_real_error} shows that the real error increases varying $n$ as the model error plays a dominant role in the real error while the expression error of Chengdu is small and the model error of MLP is large. In addition, because the space size of the Xi'an dataset is much smaller than that of Chengdu and NYC, the optimal $n$ of Xi'an is smaller than that of the other two cities.
}

\subsection{Case Study on Effect of Minimizing Real Error}
\revision{
In this section, we explore the effect of real error on two crowdsourcing problems (i.e., task assignment \cite{cheng2021queueing, tong2017flexible} and route planning \cite{wang2020demand}). We test two prediction models: Dmvst-Net and DeepST in the experiment.

\textbf{Task Assignment.} Task assignment refers to sending location-based requests to workers, based on their current positions, such as ride-hailing. We use two state-of-the-art prediction-based task assignment algorithms (i.e., LS \cite{cheng2021queueing}, POLAR \cite{tong2017flexible}) to dispatch orders under different values of $n$. The goal of LS is to maximize total revenue while the goal of POLAR is to maximize the number of served orders. Thus, we use the total revenue and order quantity as metrics for the two algorithms. We compare the performance of the two algorithms using different prediction models in this paper. Specific experimental setup in this paper is as the same as the default setting in our previous work \cite{cheng2021queueing}.

\begin{figure}[t!]\vspace{-3ex}
	\subfigure[][{\scriptsize Order Quantity}]{
		\scalebox{0.5}[0.5]{\includegraphics{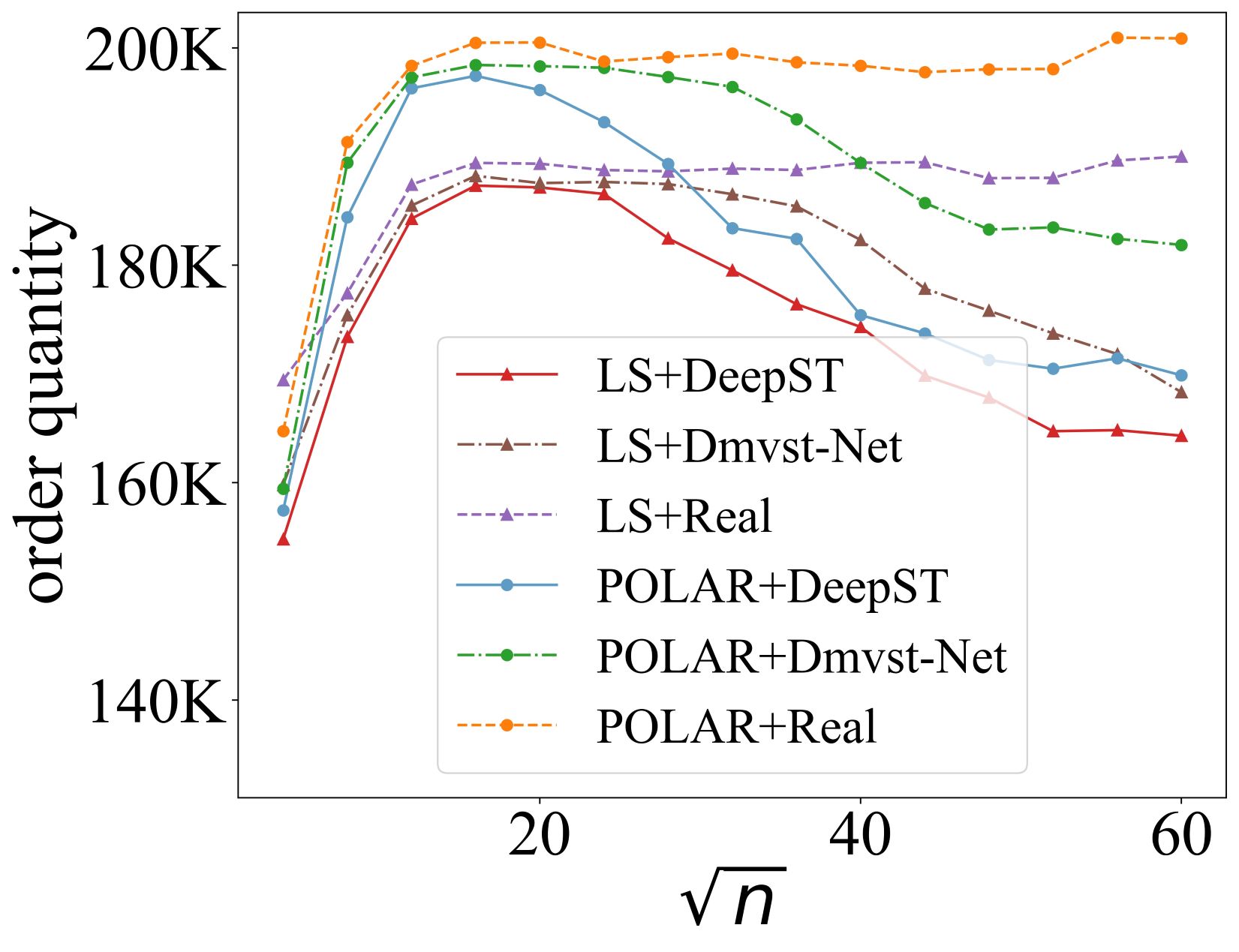}}
		\label{fig:nyc_order_quantity}}
	\subfigure[][{\scriptsize  Total Revenue}]{
		\scalebox{0.5}[0.5]{\includegraphics{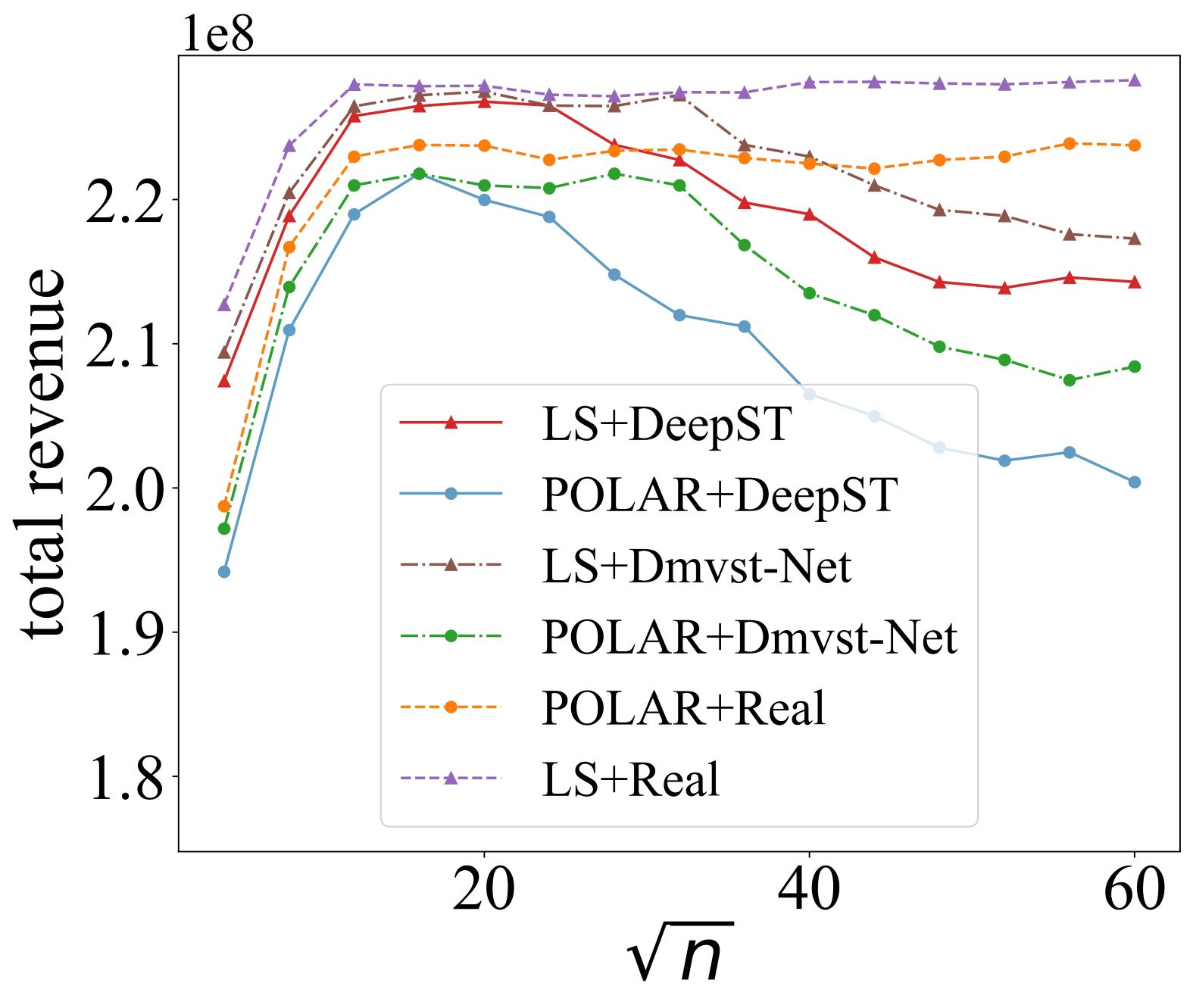}}
		\label{fig:nyc_totall_revenue}}\vspace{-1ex}
	\caption{Effect of $n$ on Task Assignment (NYC)}
	\label{fig:task_assignment_nyc}\vspace{-2ex}
\end{figure}
\begin{figure}[t!]
	\subfigure[][{\scriptsize Order Quantity}]{
		\scalebox{0.5}[0.5]{\includegraphics{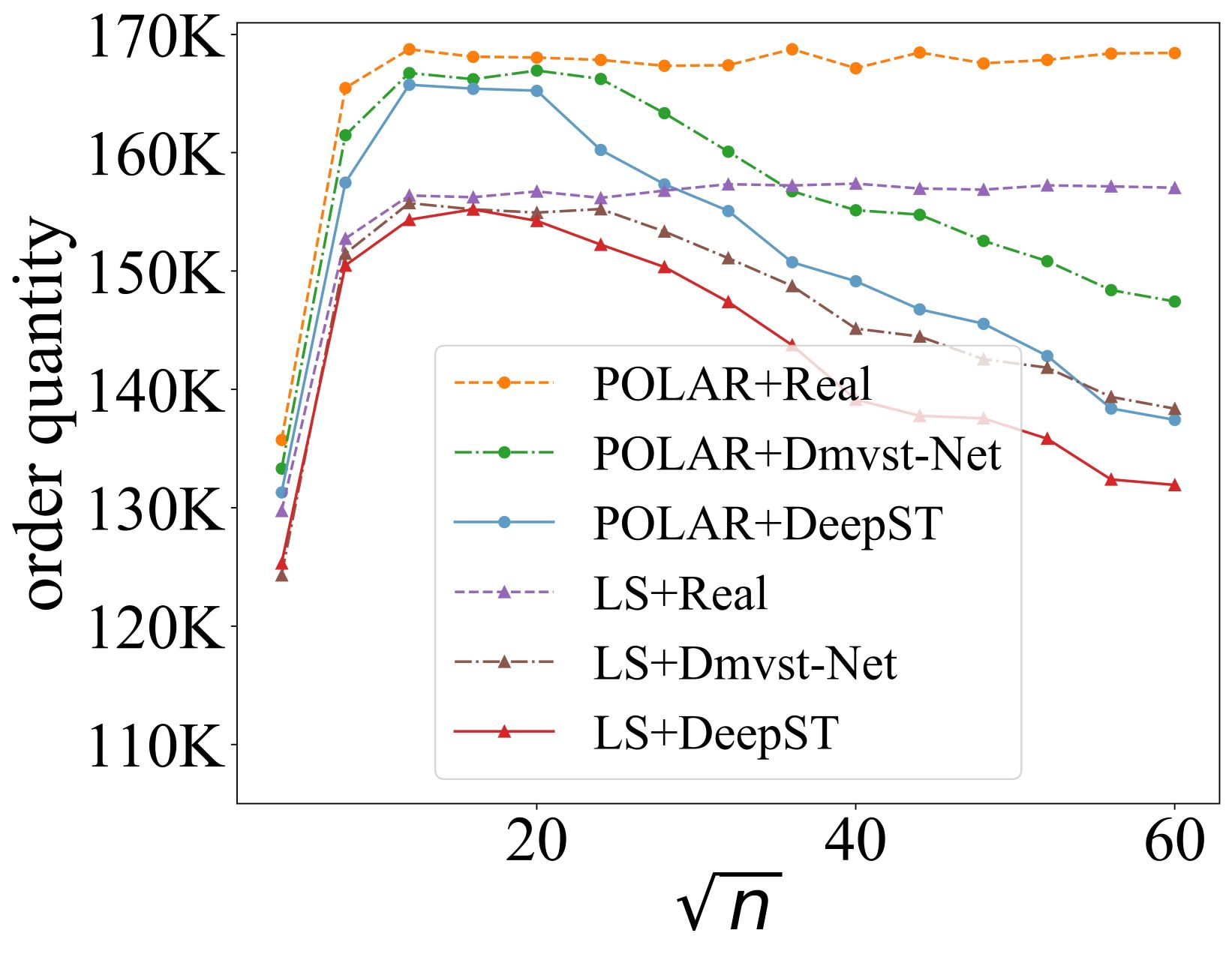}}
		\label{fig:chengdu_order_quantity}}
	\subfigure[][{\scriptsize  Total Revenue}]{
		\scalebox{0.5}[0.5]{\includegraphics{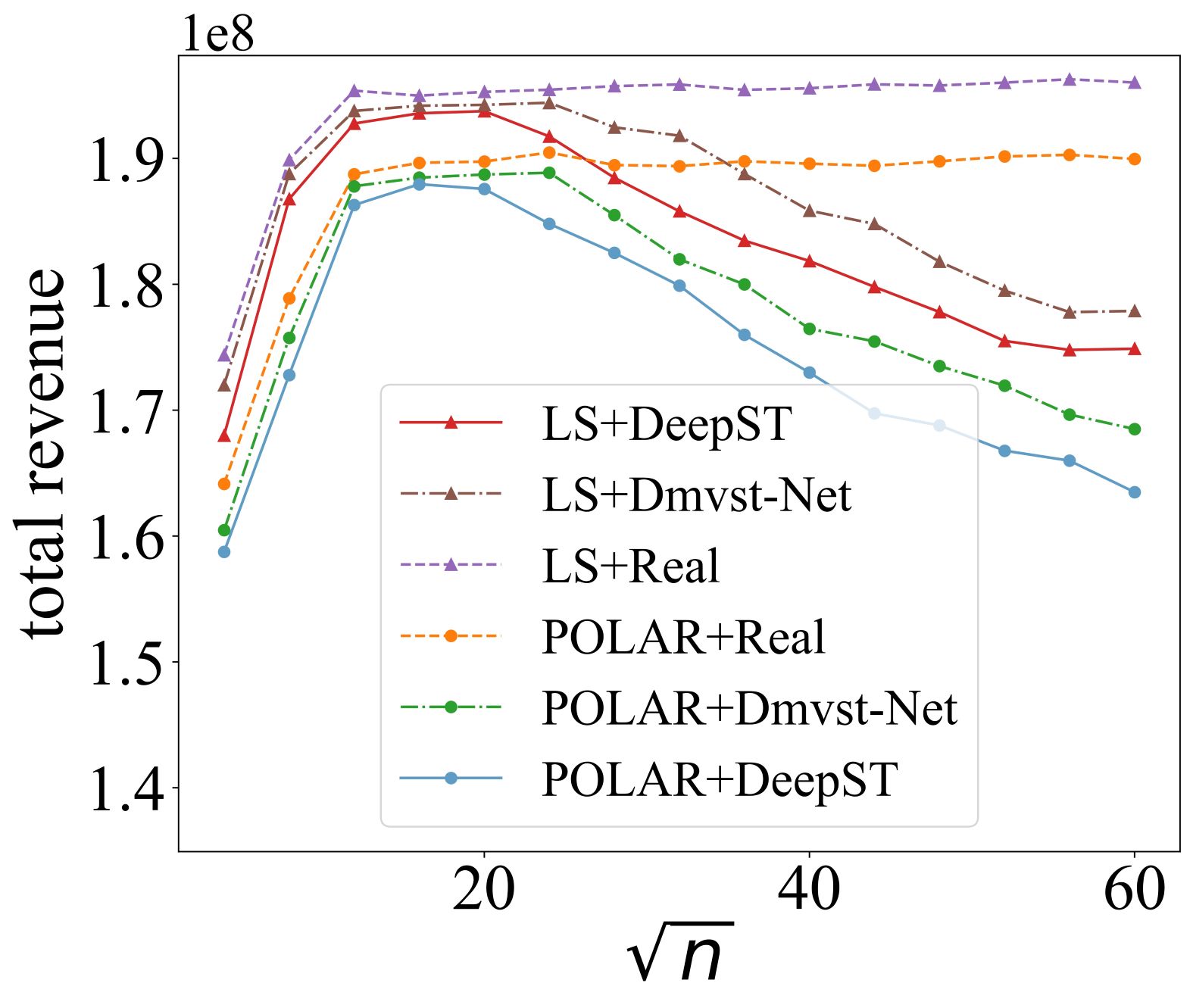}}
		\label{fig:chengdu_totall_revenue}}\vspace{-1ex}
	\caption{Effect of $n$ on Task Assignment  (Chengdu)}
	\label{fig:task_assignment_chengdu}\vspace{-3ex}
\end{figure}

Figure \ref{fig:task_assignment_nyc}$\sim$\ref{fig:task_assignment_xian} show that total revenue and order quantity of the prediction-based dispatching algorithms vary under different values of $n$. When using the predicted results, both algorithms show an increasing first and then decreasing trend in revenue, because the real error is large when $n$ is too small or too large. When POLAR and LS use real order data such that the model error becomes $0$, the real error is equivalent to the expression error. It means that the real error decreases as $n$ increases. Therefore, the performance of Polar and LS will not decrease due to the large $n$ when using the real order data, which is also consistent with the changing trend of the real error. In addition, the order distribution of Xi'an is more even than that of the other two cities because of its smaller area. Therefore, the optimal $n$ in Xi'an is less than that in the other two cities. In short, the experimental results verify that the real error is an important factor affecting the performance of the algorithms in task assignment. 

\textbf{Route Planning.} Route planning is a central issue in shared mobility applications such as ride-sharing, food delivery and crowdsourced parcel delivery. We use the state-of-the-art algorithm, DAIF \cite{wang2020demand}, to verify the effect of $n$ on route planning problem. We use the default parameters of the original paper \cite{wang2020demand} in this experiment, and take the number of served requests and the unified cost as the metrics of DAIF. Figure \ref{fig:route_planning_nyc} shows that the number of served requests of DAIF first increases then decreases when $n$ increases. The unified cost of DAIF is minimized when $n=16\times 16$. Using the actual number of orders, DAIF gets better performance with a large $n$. Although route planning problem is affected less by real error compared with task assignment problem, the size of grid affects the performance of prediction-based algorithms.

\begin{figure}[t!]\vspace{-3ex}
	\subfigure[][{\scriptsize Order Quantity}]{
		\scalebox{0.5}[0.5]{\includegraphics{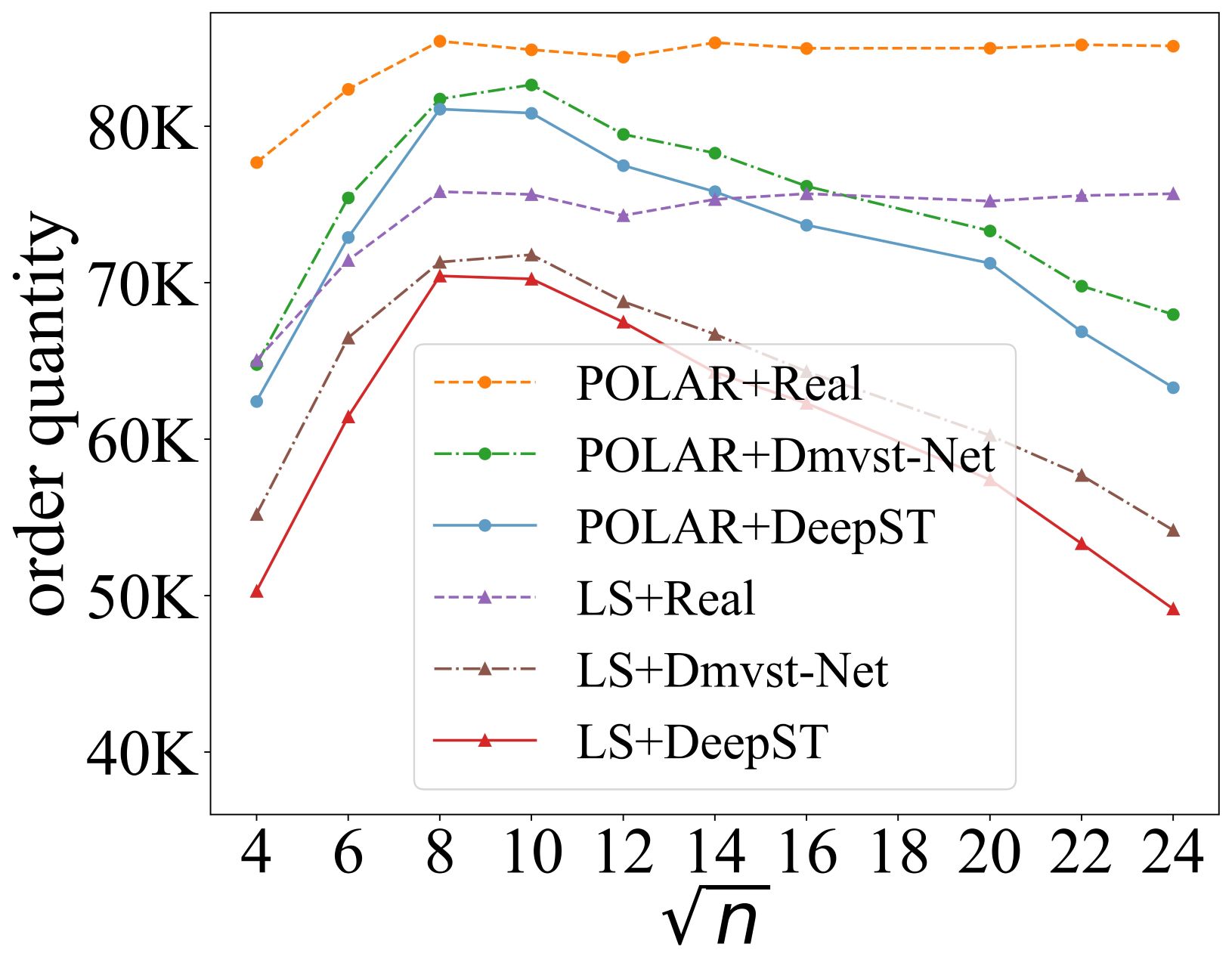}}
		\label{fig:xian_order_quantity}}
	\subfigure[][{\scriptsize Total Revenue}]{
		\scalebox{0.5}[0.5]{\includegraphics{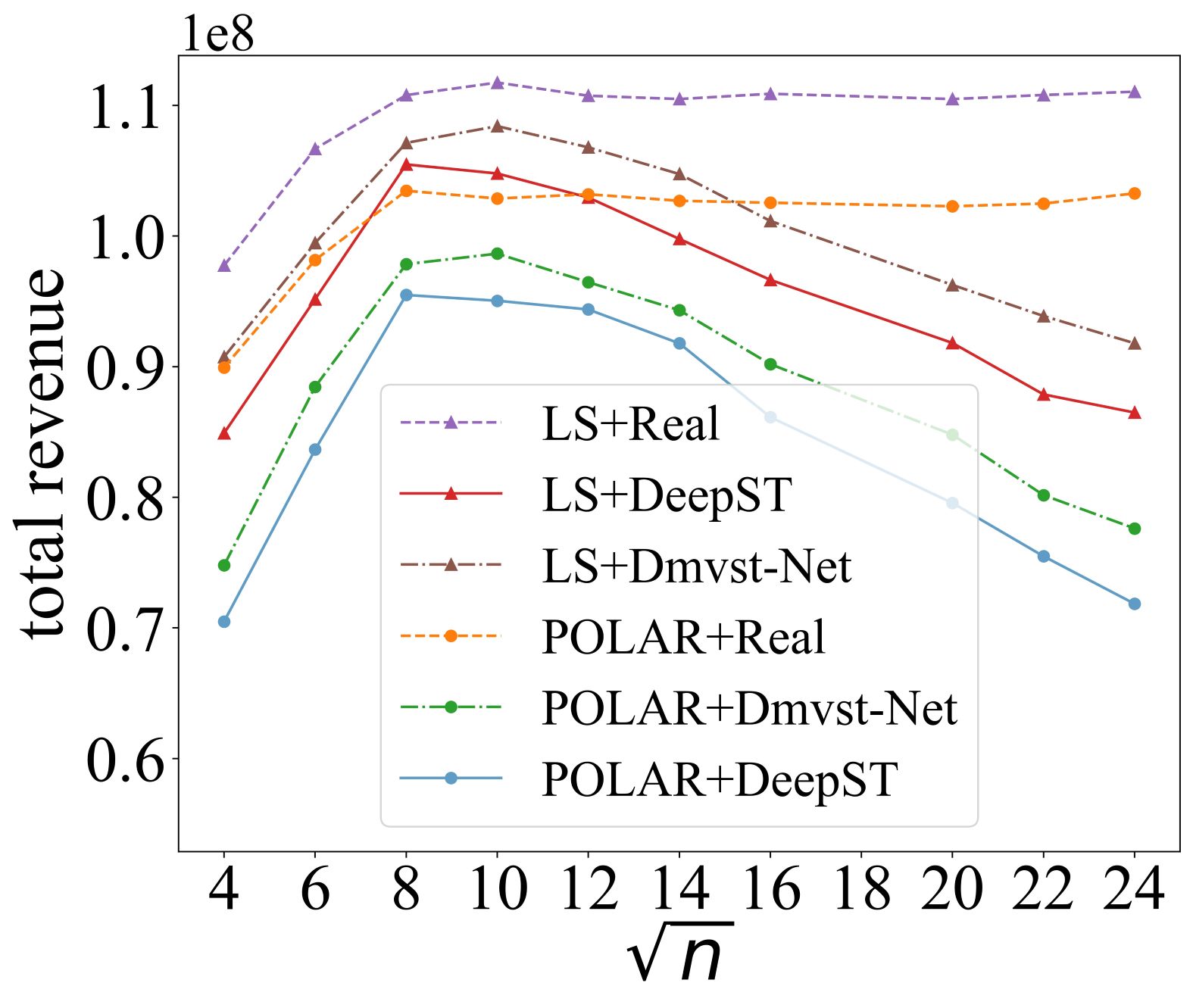}}
		\label{fig:xian_totall_revenue}}\vspace{-1ex}
	\caption{Effect of $n$ on Task Assignment (Xi'an)}
	\label{fig:task_assignment_xian}\vspace{-2ex}
\end{figure}
\begin{figure}[t!]
	\subfigure[][{\scriptsize Served Requests}]{
		\scalebox{0.495}[0.495]{\includegraphics{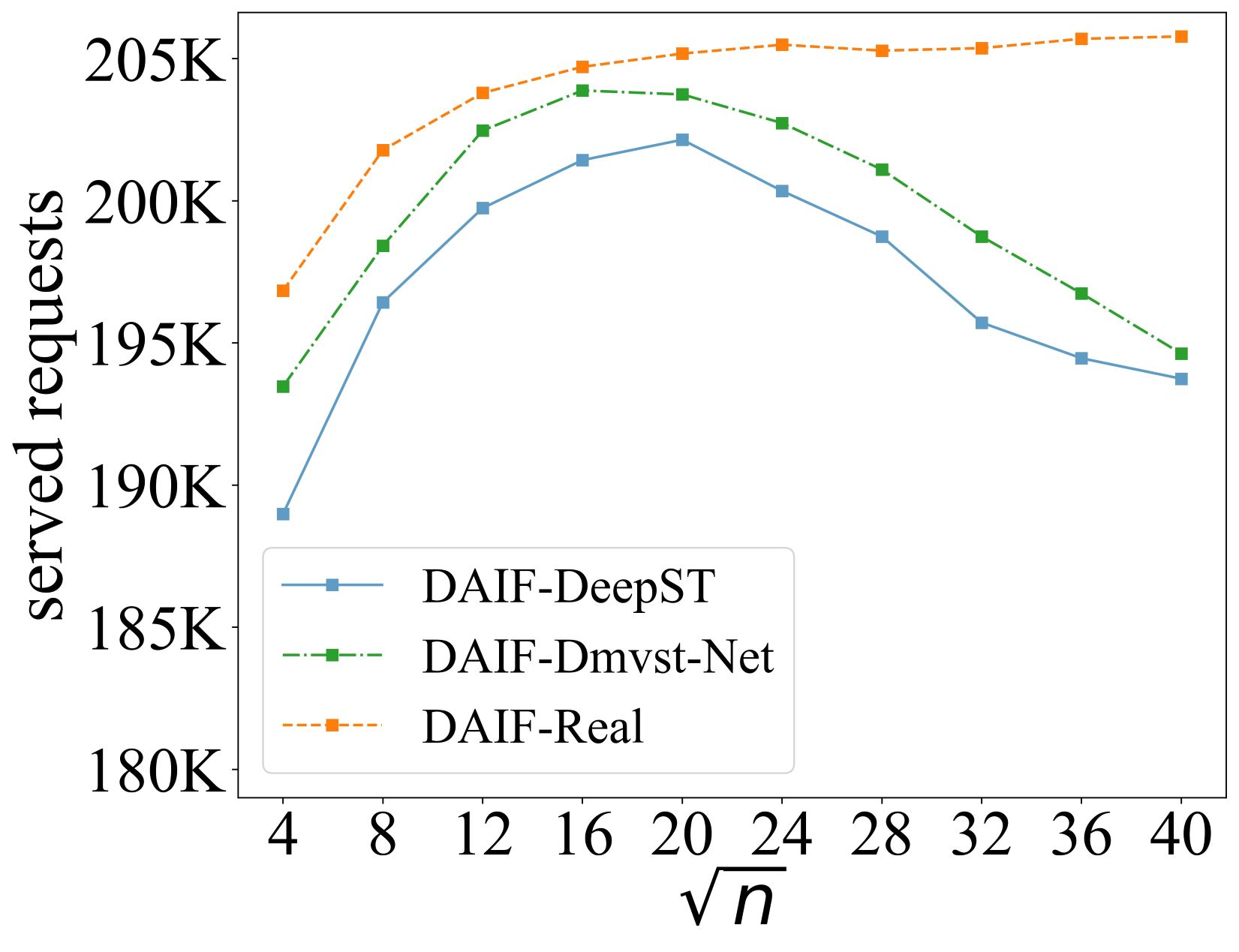}}
		\label{fig:nyc_order_quantity_route}}
	\subfigure[][{\scriptsize Unified Cost}]{
		\scalebox{0.495}[0.495]{\includegraphics{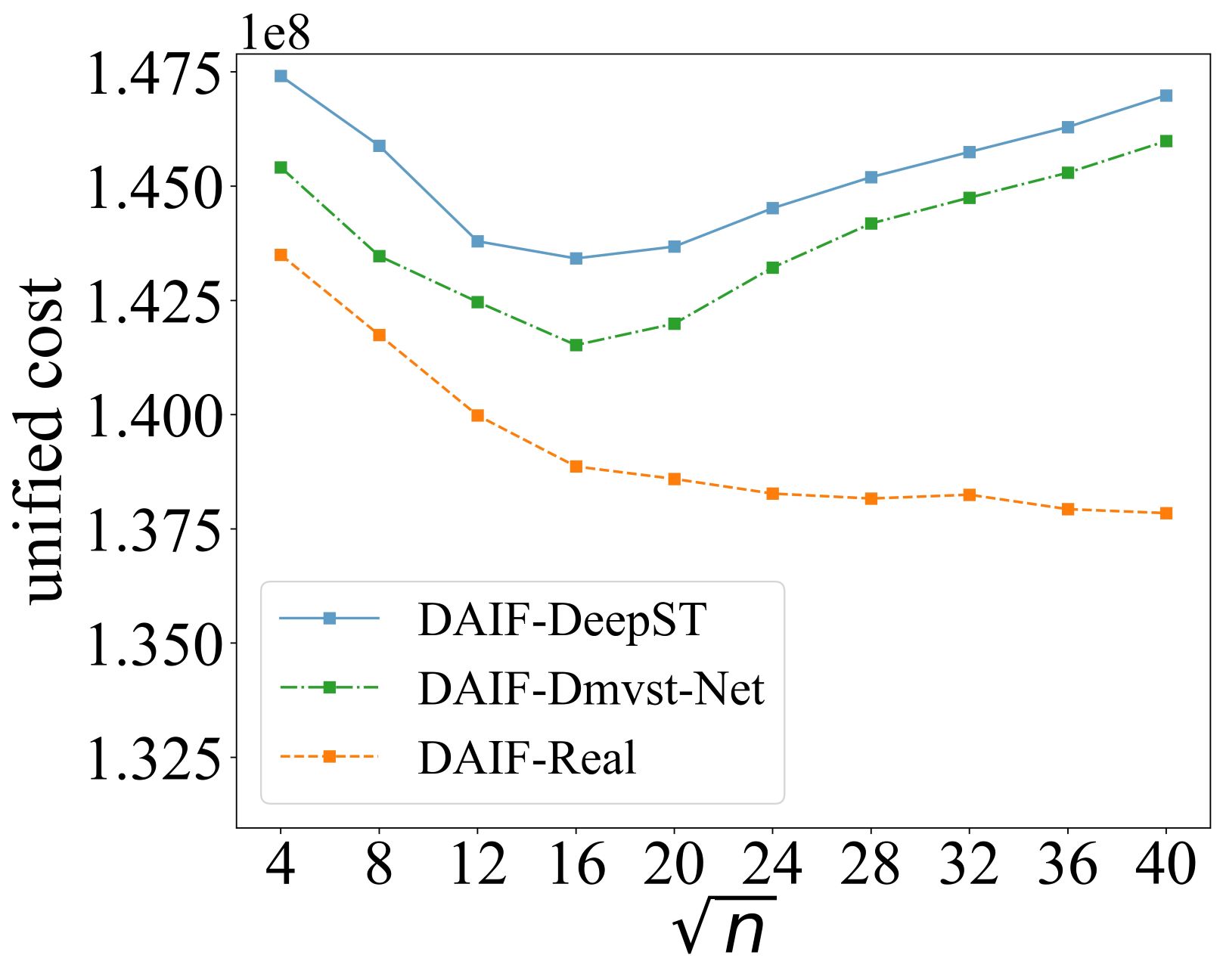}}
		\label{fig:nyc_totall_revenue_route}}\vspace{-1ex}
	\caption{Effect of $n$ on Route Planning (NYC)}
	\label{fig:route_planning_nyc}\vspace{-3ex}
\end{figure}

Table \ref{tab:promotion} shows the improvement of the original algorithm by selecting the optimal grid size with DeepST as the prediction model on NYC. Original $n$ represents the default value of $n$ set in \cite{tong2017flexible, cheng2021queueing, wang2020demand}, while optimal $n$ denotes the optimal grid size found by our GridTuner. The results show that both POLAR and DAIF can achieve performance gains with the optimal grid size. Due to the selection of the default $n$ in the existing paper \cite{tong2017flexible} is close to the optimal $n$, the performance of LS has no obvious improvement.}

\begin{table}[h!]
	\centering 
	{\small\scriptsize
		\caption{\small Promotion of the prediction-based algorithms}
		\label{tab:promotion}
		\begin{tabular}{l|l|lll}
			Metric&Algorithm&Optimal $n$&Original $n$&Improve ratio\\
			\hline
			{Served Order Number}&POLAR&$16\times 16$&$50\times 50$&$13.6\%$\\
			{Total Revenue}&POLAR&$16\times 16$&$50\times 50$&$8.97\%$\\
			\hline
			{Total Revenue}&LS&$20\times 20$&$16\times 16$&$0.13\%$\\
			{Served Order Number}&LS&$20\times 20$&$16\times 16$&$0.7\%$\\
			\hline
			{Unified Cost}&DAIF&$16\times 16$&$12\times 12$&$0.76\%$\\
			{Served Requests}&DAIF&$20\times 20$&$12\times 12$&$3.35\%$\\
			\hline
			\hline
		\end{tabular}
	}\vspace{-1ex}
\end{table}

\subsection{Experiment Result of Optimization Searching Algorithms}

Since the mean of the event quantity in the same grid varies in different periods of a day, the expression error of each time slot is different, leading to the different optimal solutions of each time slot. 
\revision{
In this section, we use the algorithms proposed in Section \ref{sec:solution2} to calculate the optimal partition scheme of different cities and compare the performance of them with the \textbf{Brute-force Search} (i.e., traverses all the values to find the optimal $n$).
We use three indicators to measure the quality of the solution found by the algorithm and the efficiency of the algorithm: \textbf{cost} denotes the cost of time; \textbf{probability} denotes the probability of obtaining the optimal solution (i.e., the number of finding out the optimal solution divided by the number of time slots); \textbf{optimal ratio ($OR$)} is denoted as
$
OR=\frac{o_a}{o_r}
$,
where $o_r$ denotes the optimal order count served by driver while using the POLAR as dispatching algorithm in NYC and $o_a$ represents the results optimized by the algorithm.
}

\revision{
	\begin{table}[t!]
		\centering
		{
			\caption{\small Performance of the algorithms.} \label{table3}
			\begin{tabular}{l|l|ccc}
				{\bf City} & {\bf Algorithm} & {\bf Cost (h)}&{\bf Probability}&{\bf OR} \\ \hline
				{NYC}&Ternary Search&7.03&52.08$\%$&97.83$\%$\\
				{NYC}&Iterative Method&{\bf5.58}&{\bf81.25$\%$}&{\bf98.77}$\%$\\
				{NYC}&Brute-force Search&47.43&100.00$\%$&100.00$\%$\\\hline
				{Chengdu}&Ternary Search&6.32&70.83$\%$&98.35$\%$\\
				{Chengdu}&Iterative Method&4.53&{\bf95.83$\%$}&{\bf99.77$\%$}\\
				{Chengdu}&Brute-force Search&43.26&100.00$\%$&100.00$\%$\\\hline
				{Xi'an}&Ternary Search&{\bf3.90}&60.42$\%$&97.98$\%$\\
				{Xi'an}&Iterative Method&$3.31$&{\bf91.67$\%$}&{\bf99.57$\%$}\\
				{Xi'an}&Brute-force Search&$21.76$&100.00$\%$&100.00$\%$\\\hline\hline
			\end{tabular}
		}\vspace{-2ex}
	\end{table}
	
The experimental results in Table \ref{table3} show that Ternary Search and Iterative Method both can greatly reduce the time cost of finding the optimal solution compared with the Brute-force Search. Meanwhile, both algorithms can find the global optimal solution with high probabilities. With the reasonable choice of bound and initial position of Iterative Method, its execution efficiency and probability of finding the optimal solution are better than that of Ternary Search. According to Table \ref{table3}, sub-optimal solutions achieved by Ternary Search are at most 3$\%$ less than the optimal results (and 1.5$\%$ for Iterative Method), which shows the effectiveness of our grid size selection algorithms. 
}

\noindent\textbf{Summary:}
The experimental results show that the larger real error often leads to the decrease of the payoff of the dispatching algorithms. At the same time, Ternary Search and Iterative Method proposed in this paper can effectively find the optimal solution to the OGSS by minimizing the upper bound of the real error. Furthermore, this paper improves the effect of the original algorithm by selecting a reasonable size $n$. Specially, the performance of POLAR improves by $13.6\%$ on served order number and $8.97\%$ on total revenue. Finally, this paper also studies the influence of different traffic prediction algorithms on the optimal size of MGrids. The results show that when the accuracy of the prediction algorithm is high, the whole space can be divided into more MGrids to reduce the expression error. On the contrary, when the accuracy of the prediction algorithm is low, we need to make the area of a MGrid larger to reduce the model error.

\section{Related Work}
\label{sec:related}
In recent years, with the rise of various online taxi-hailing platforms, more and more researchers have been working on how to assign tasks to workers. Summarized the publishing models in \cite{kazemi2012geocrowd}, the task assignment problem mainly is classified in two modes: worker selected task and server assigned task. Compared with the former, the latter is easier to find the optimal global solution, and the major online taxi-hailing platforms mainly adopt the latter, which attracts more and more researchers' attention.

There are two main modes of order distribution on the platform: online \cite{tong2017flexible, cheng2019queueing, tong2016online, wang2020demand, asghari2018adapt} and the other is offline \cite{zheng2018order, ma2013t, thangaraj2017xhare, chen2018price}. The online task assignment faces more tremendous challenges than the offline task assignment due to the lack of follow-up order information. However, the emergence of traffic prediction technology \cite{zhang2017deep, yao2018deep, guo2019attention} has solved this problem well. With the continuous improvement of these traffic forecasting work, the demand-aware algorithm \cite{tong2017flexible, cheng2019queueing, wang2020demand, zhao2020predictive, asghari2018adapt} for task assignment has more advantages than some traditional algorithm \cite{kazemi2012geocrowd, karp1990optimal, chen2018price, zheng2018order}. The optimal solution of task assignment based on the supply and demand can be approximately equivalent to the optimal solution of offline task assignment problem when the prediction result of the traffic prediction algorithm is close to the real.

Several traffic prediction methods \cite{zhang2017deep, yao2018deep, guo2019attention, he2015high, cressie2015statistics} divides the entire space into grids based on latitude and longitude and then predicts the number of orders in each region. The residual network is introduced into the traffic prediction in \cite{zhang2017deep} so that the deep neural network can better reduce the deviation between the predicted results and the actual results. \cite{yao2018deep} tries to combine different perspectives to predict future order data, including time perspective, space perspective, and semantic perspective. The results show that the multi-view spatiotemporal network can improve the prediction performance of the model. In addition, an attention mechanism is introduced in \cite{guo2019attention} to mine the dynamic spatiotemporal correlation of traffic data to optimize the prediction performance.

Combining the result of predicted future distribution of tasks, algorithms for task assignment can better solve the problem. The work \cite{cheng2019queueing} proposed a framework based on queuing theory to guide the platform for order dispatching, which used queuing theory combined with the distribution of future orders and drivers in the region to predict the waiting time of drivers before they received the next order after sending the current order to the destination. In addition, a two-stage dispatching model is proposed in \cite{tong2017flexible}. In the first stage, the platform will pre-assign drivers based on the predicted number of regional orders and direct them to the likely location of the orders, while in the second stage, it will assign the actual orders.

The accuracy of traffic prediction will significantly affect the performance of this algorithm. However, The order dispatching algorithms pay attention to the model error and the expression error caused by the uneven distribution of orders in the grids. Our problem mainly focuses on how to divide model grid to balance the model error and the expression error to improve the effectiveness of the order dispatching algorithm based on supply and demand prediction.

\section{Conclusion}
\label{sec:conclusion}

In this paper, we propose a more fine-grained measure of prediction bias, namely real error, and investigate how to minimize it. The real error is mainly composed of expression error and model error. Expression error is caused by using the order quantity of large regions to estimate the number of spatial events of HGrids, while  model error  is the inner error of the prediction model. We show that the summation of the expression error and the model error is the upper bound on the real error. We solve expression error and model error and analyze the relationship between them and MGrids. Through the above analysis, we propose two algorithms to minimize the real error as much as possible by minimizing its upper bound. Finally, we verify the effectiveness of our algorithm through experiments and analyze the role of real error for spatiotemporal prediction models.


\newpage

\bgroup\small
\bibliographystyle{ieeetr}
\let\xxx=\bibitem\def\bibitem{\par\vspace{1mm}\xxx}
\bibliography{../references/add}
\egroup

\appendix

\subsection{Distributions of Order Datasets}
\label{append:distribution_orders}
Figure \ref{fig:order_distribution} shows the distribution of orders in test data set from 8:00 A.M. to 8:30 A.M.

\begin{figure*}[t!]
	\centering
	\subfigure[][{\scriptsize NYC}]{
		\scalebox{0.65}[0.65]{\includegraphics{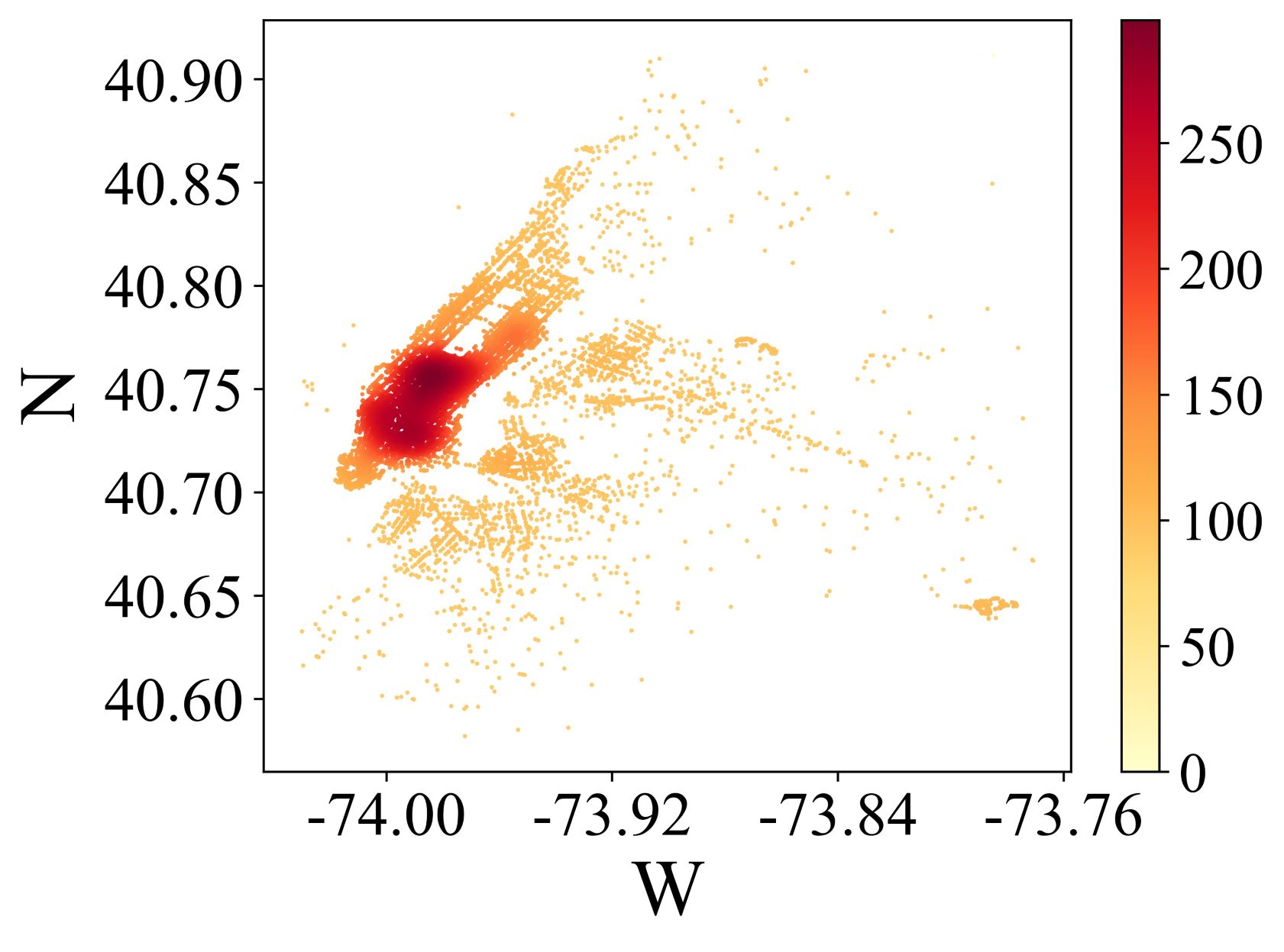}}
		\label{fig:nyc_distribution}}
	\subfigure[][{\scriptsize Chengdu}]{
		\scalebox{0.65}[0.65]{\includegraphics{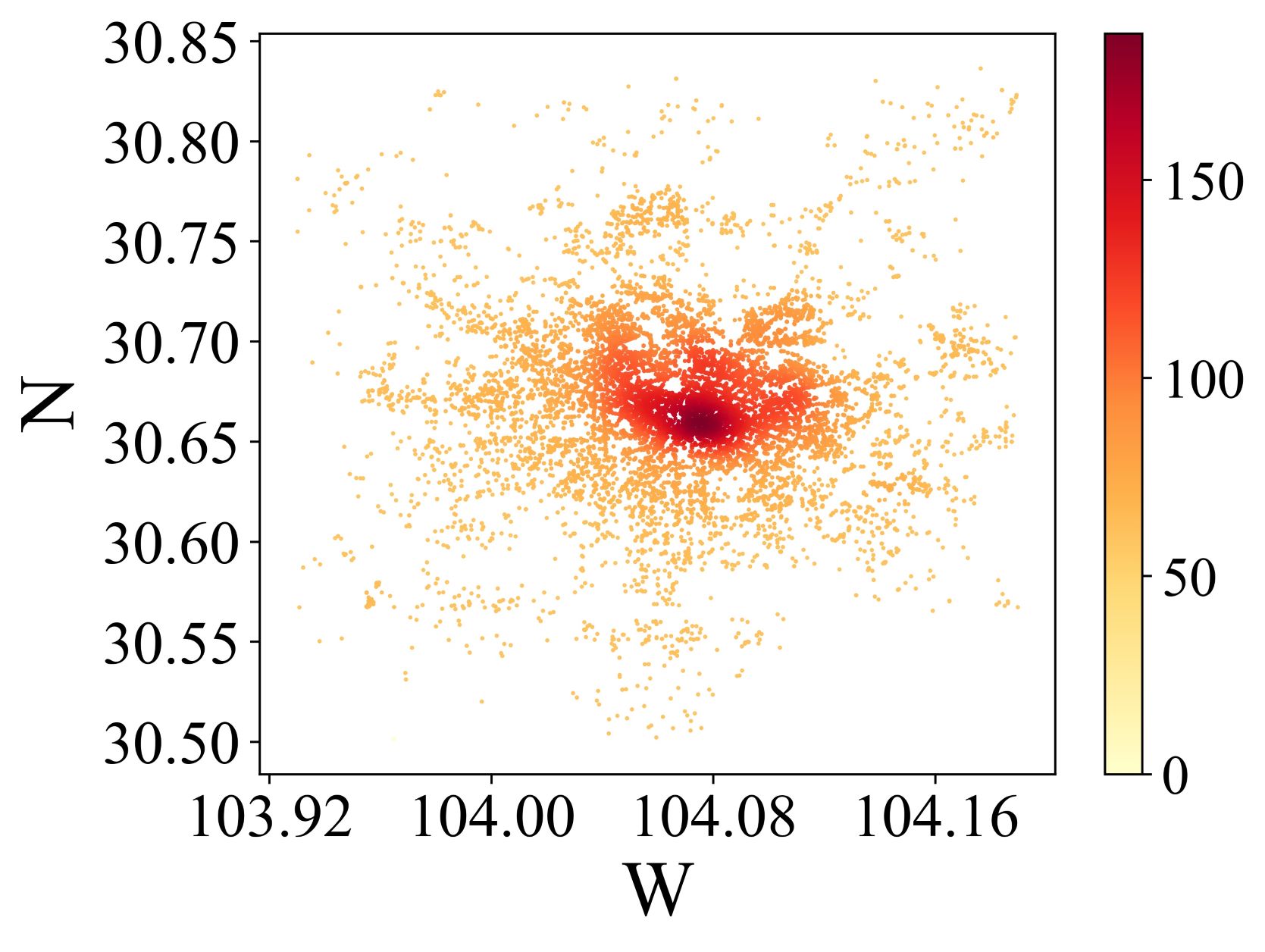}}
		\label{fig:chengdu_distribution}}
	\subfigure[][{\scriptsize Xi'an}]{
		\scalebox{0.65}[0.65]{\includegraphics{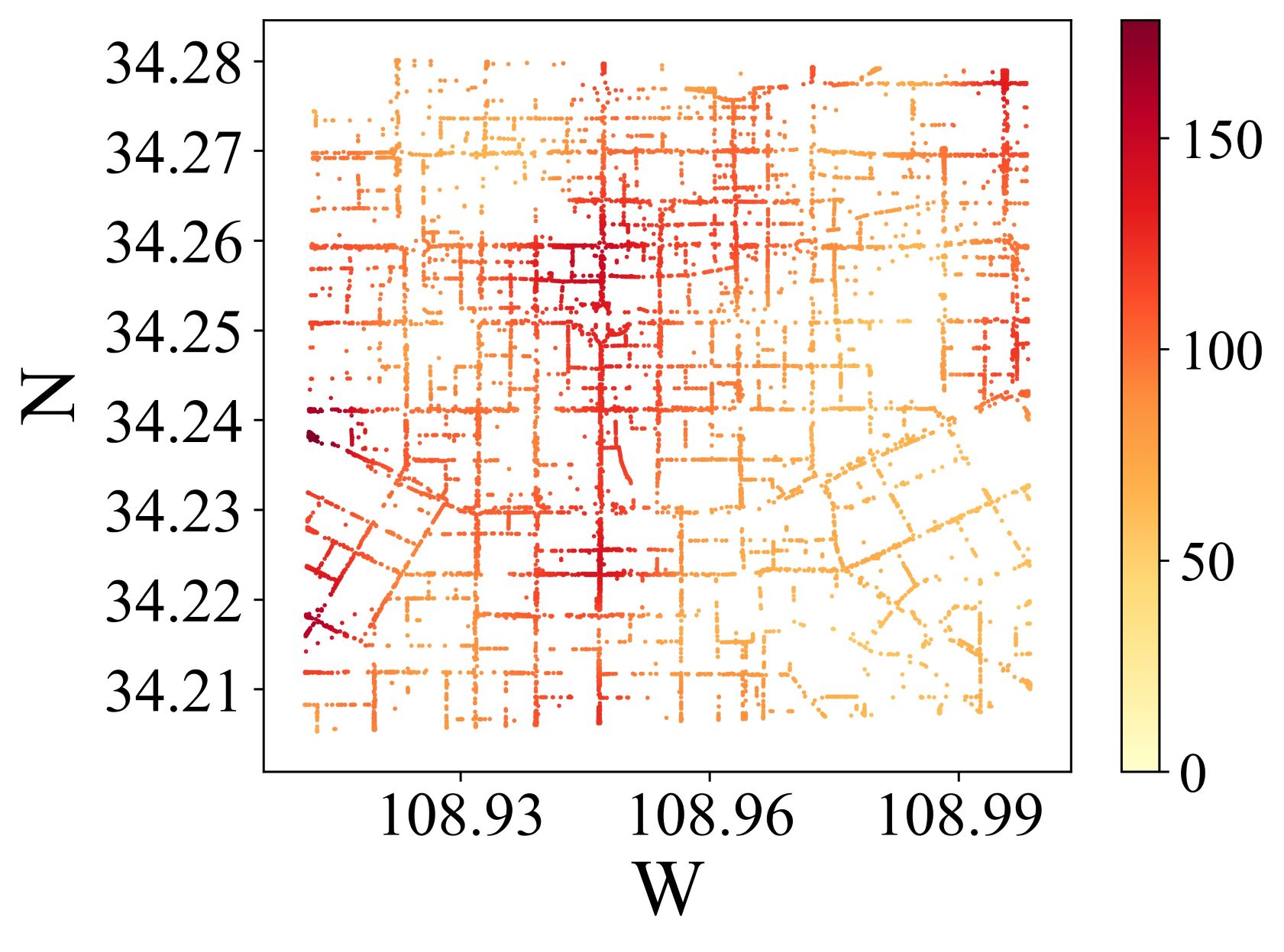}}
		\label{fig:xian_distribution}}
	\caption{Order Distributions in NYC, Chengdu and Xi'an}
	\label{fig:order_distribution}
\end{figure*}

The number of trips in the testing day is: 282,255 for NYC dataset, 238,868 for Chengdu, and 109,753 for Xi'an. We also analyze the distribution of the length of the trips in three cities as shown in Figure \ref{fig:trip_length_distribution}.
The lengths of trips in Chengdu are generally evenly distributed, but the number of long trips (longer than 45 km) is more than 1,000. The taxi trips in NYC mainly happened in Manhattan district, thus most trips are shorter than 15 km. For Xi'an dataset, since the spatial area is relatively small, most trips are shorter than 10 km.

\begin{figure*}[t!]
	\subfigure[][{\scriptsize Chengdu}]{
		\scalebox{0.35}[0.35]{\includegraphics{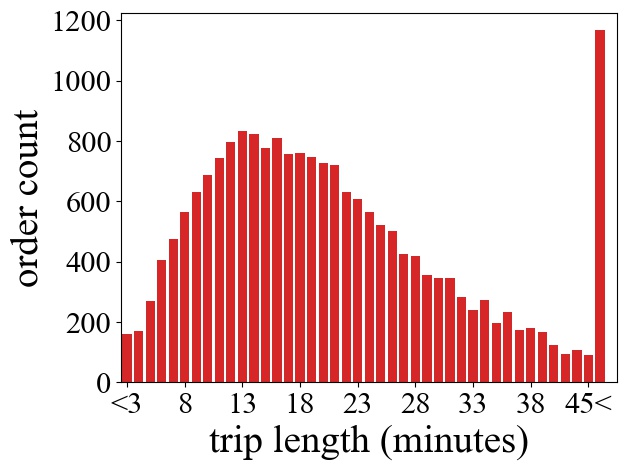}}
		\label{fig:chengdu_trip_length_distribution}}
	\subfigure[][{\scriptsize NYC}]{
		\scalebox{0.35}[0.35]{\includegraphics{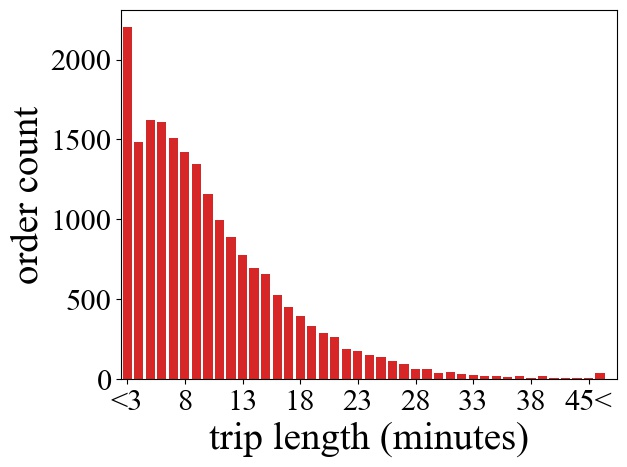}}
	\label{fig:nyc_trip_length_distribution}}
	\subfigure[][{\scriptsize Xi'an}]{
		\scalebox{0.35}[0.35]{\includegraphics{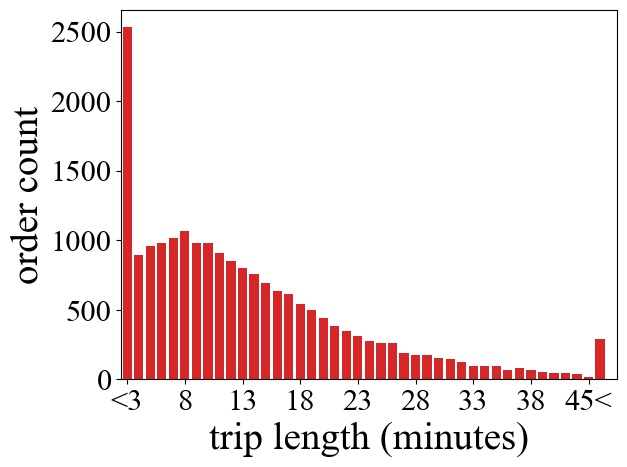}}
		\label{fig:xian_trip_length_distribution}}
	\caption{\small Distribution of Trip Length in Different Cities}
	\label{fig:trip_length_distribution}
\end{figure*}

\subsection{Relationship between Expression Error and the Uniformity of the Distribution}
\label{sec:realation_uniformity}
Expression error refers to the error caused by estimating the number of events $\hat{\lambda}_{ij}$ in a HGrid $r_{ij}$ with the prediction result $\hat{\lambda}_i$ of the MGrid $r_{i}$. The uniformity of the distribution of events within the MGrid will lead to a large expression error. We use $D_\alpha\left(N\right)$ to represent the degree of unevenness in the distribution of events within a MGrid. In the experiment, we set the parameter $m$ as $8\times 8$ with $n=16\times 16$, which means each MGrid has an area of $3.7578$ $km^2$. Then we calculate the imbalance $D_\alpha\left(64\right)$  of event distribution within each MGrid and the summation of the expression error $E\left(i,j\right)$ of each HGrid $r_{ij}$ in the MGrid $r_{i}$. Then, we plot the corresponding relationship between them into a scatter diagram.

Figure \ref{fig:d_alpha_e} shows that the expression error gradually increases with the increase of imbalance of event distribution within a MGrid. On the other hand, we find that many MGrids has $D_\alpha\left(64\right)<10$ because events are unevenly distributed in New York City, leading to the scarcity of events in several grids.

The distribution of events in two different MGrids shown in Figure \ref{fig:example2} where each point denotes a spatial event. The event distribution shown in Figure \ref{subfig:sub_example4} is uneven. There is a large empty area in the upper left corner, and there is a long main road in the middle with lots of events. On the contrary, the event distribution in Figure \ref{subfig:sub_example3} is more uniform. The $D_\alpha\left(64\right)$ of the right grid is $25.87$ with the expression error of $39.90$, while the $D_\alpha\left(64\right)$ of the left grid is $8.47$ with the expression error of $8.8$.
\begin{figure}[t!]
	\subfigure[][{\scriptsize a Grid with  $D_\alpha=25.87$}]{
		\scalebox{0.5}[0.5]{\includegraphics{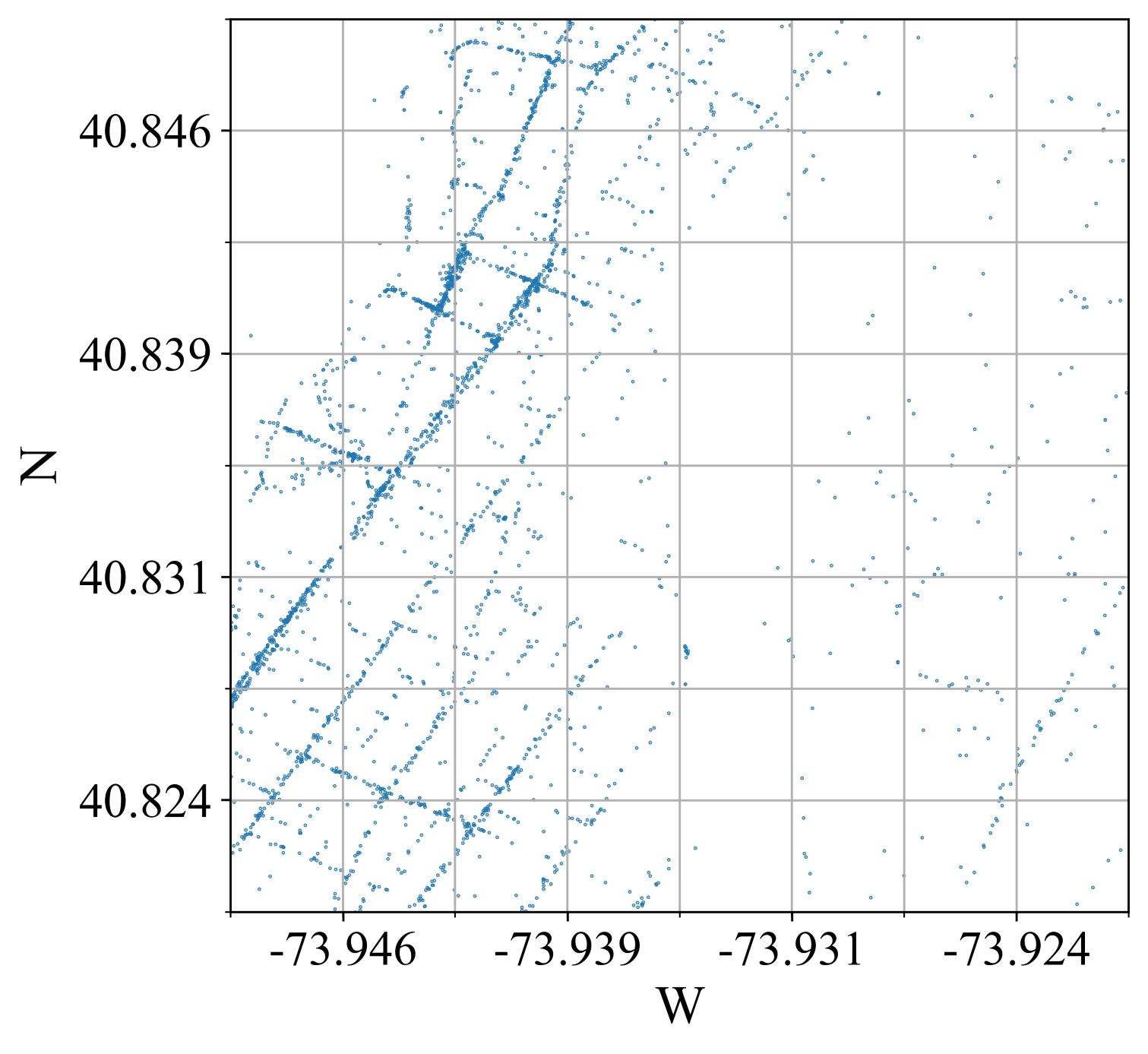}}
		\label{subfig:sub_example4}}
	\subfigure[][{\scriptsize a Grid with  $D_\alpha=8.8$}]{
		\scalebox{0.5}[0.5]{\includegraphics{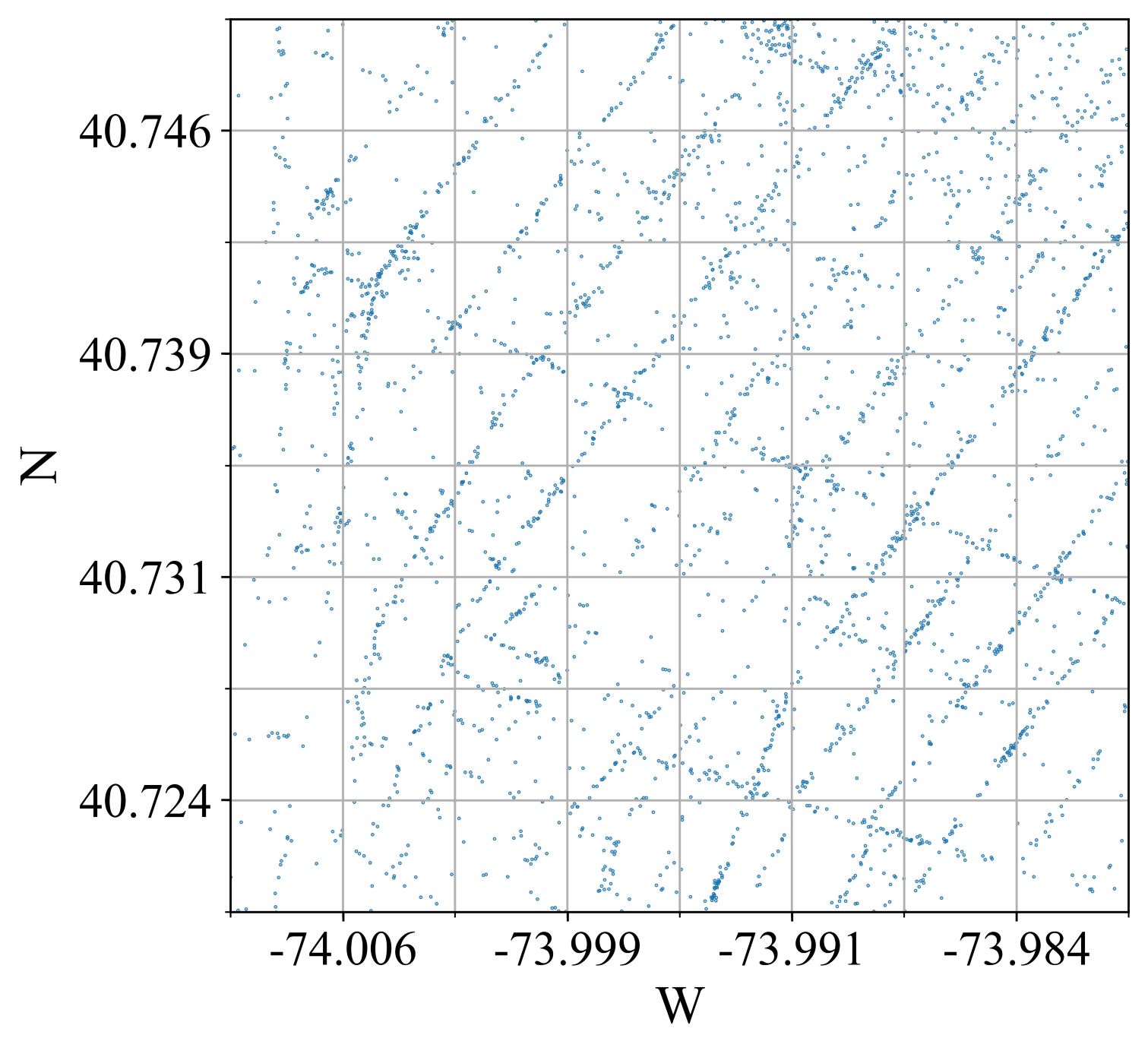}}
		\label{subfig:sub_example3}}	
	\caption{\small Two instances of event distribution}
	\label{fig:example2}
\end{figure}
\begin{figure}[t!]
    \centering
    \scalebox{0.5}[0.5]{\includegraphics{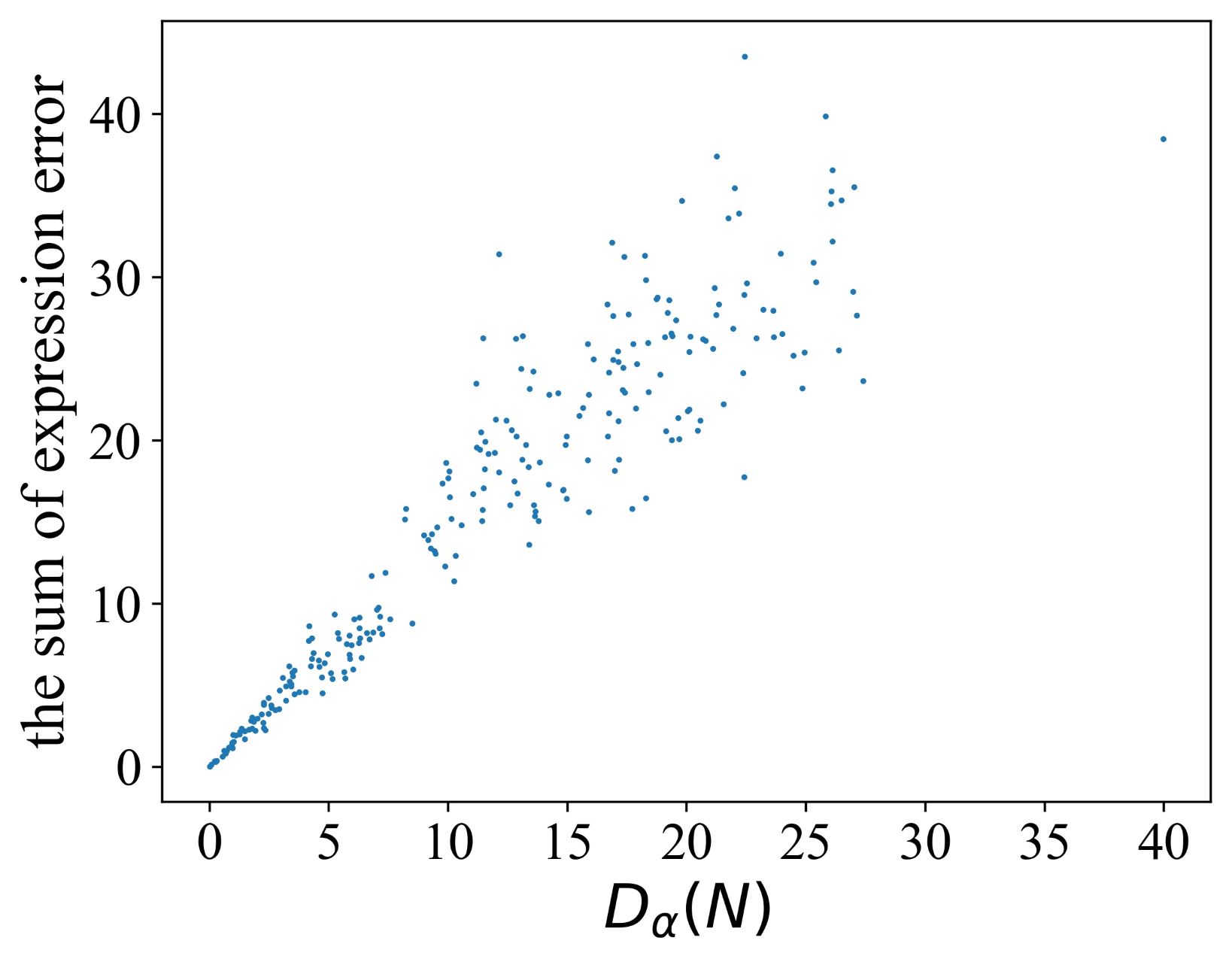}}
    \caption{\small Effect of $D_{\alpha}\left(N\right)$ on $E_e\left(i,j\right)$}
    \label{fig:d_alpha_e}
\end{figure}
\subsection{Results on HGrid Division}
The real error is defined in the HGrids. Thus, how to divide the whole space into HGrids becomes a crucial problem. We cannot guarantee that events are evenly distributed within a HGrid while the area of each HGrid is enormous. On the other hand, the calculation complexity of the expression error will be too high if the area of each HGrid is tiny.

We calculate $D_\alpha\left(N\right)$ over the whole grids with different values of $N$ based on Equation \ref{eq:d_alpha}. Then we plot the change in $D_\alpha\left(N\right)$ with respect to $N$ in two different way to estimate $\alpha_{ij}$.

\revision{
Figure \ref{fig:n_d_alpha} shows that $D_\alpha\left(N\right)$ increases with $N$ and the growth rate of $D_{\alpha}\left(N\right)$ slows down while $N$ is greater than a turning point (i.e., approximately $76$), which indicates that the distribution of events within the HGrid is uniform when $N$ is greater than $76\times 76$. However, $D_{\alpha}\left(N\right)$ continues to grow \revision{quickly} when $N$ is greater than $76\times 76$ due to the estimation of $\alpha_{ij}$ with the sample of a week or three months. Subsequent increases in $D_{\alpha}\left(N\right)$ are mainly attributed to the inaccurate estimation of the value of $\alpha_{ij}$.
}

We set $n$ to 16 and keep increasing $m$ to figure out the influence of $N$ on real error, model error and expression error. Figure \ref{fig:m_error} shows that the real error and the expression error increase as $m$ increases. \revision{When $m$ is large, the area of grids decreases, leading to the inaccurate estimation of the value of $\alpha_{ij}$. Thus, the expression error and the real error continually increase. In this paper, we want to reduce the influence of such inaccuracy on the expression error, and make the expression error to reflect the uneven distribution of events. As a result, we set $N=128\times128$ in our experiments, which allows us to reduce the influence of the inaccurate estimation of the value of $\alpha_{ij}$ and guarantee homogeneous for HGrids.}
\begin{figure}[t!]\centering
	\scalebox{0.30}[0.30]{\includegraphics{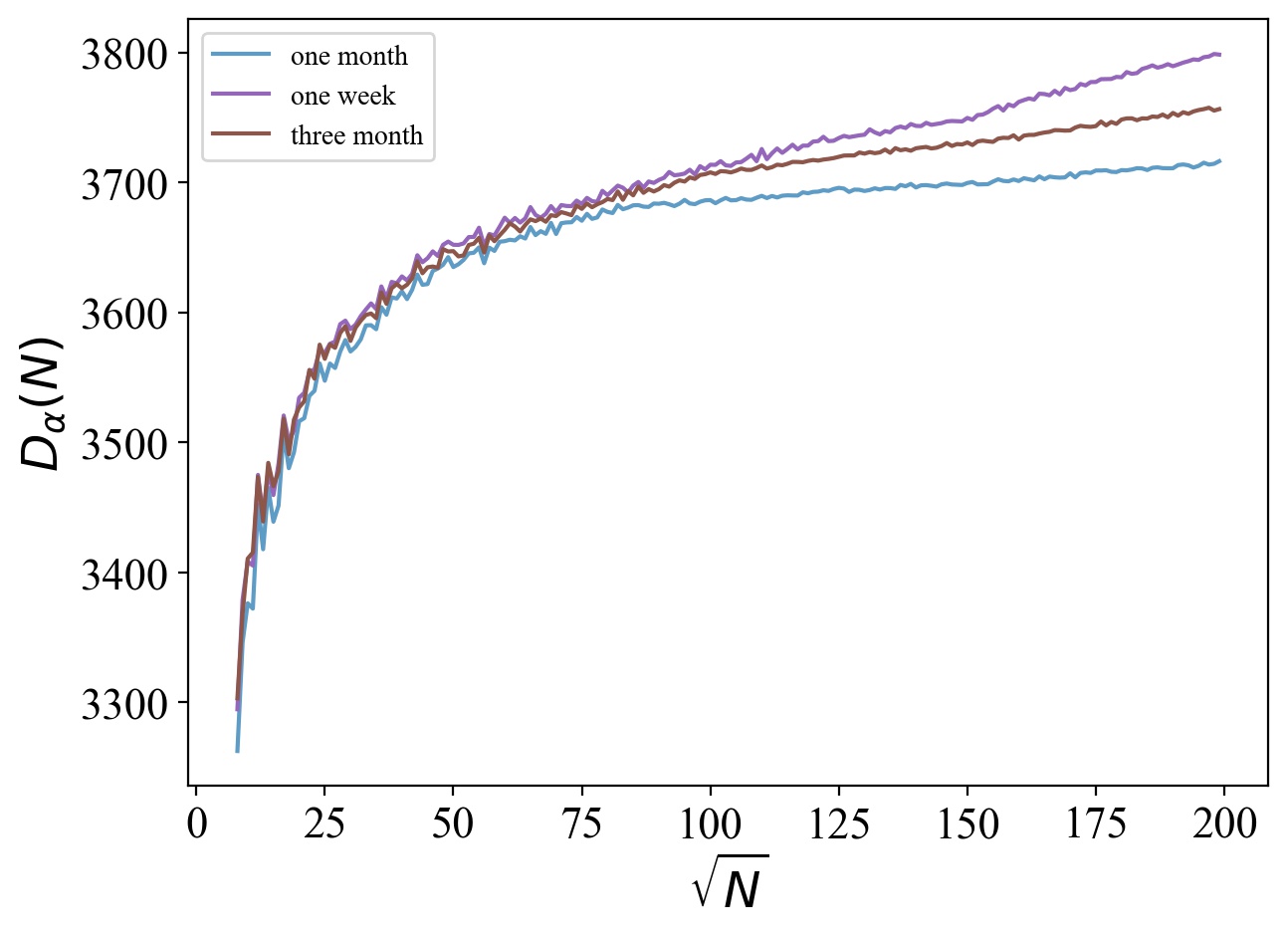}}
	\caption{\small Effect of $N$ on $D_{\alpha}\left(N\right)$}
	\label{fig:n_d_alpha}
\end{figure}
\begin{figure}[t!]
	\scalebox{0.65}[0.65]{\includegraphics{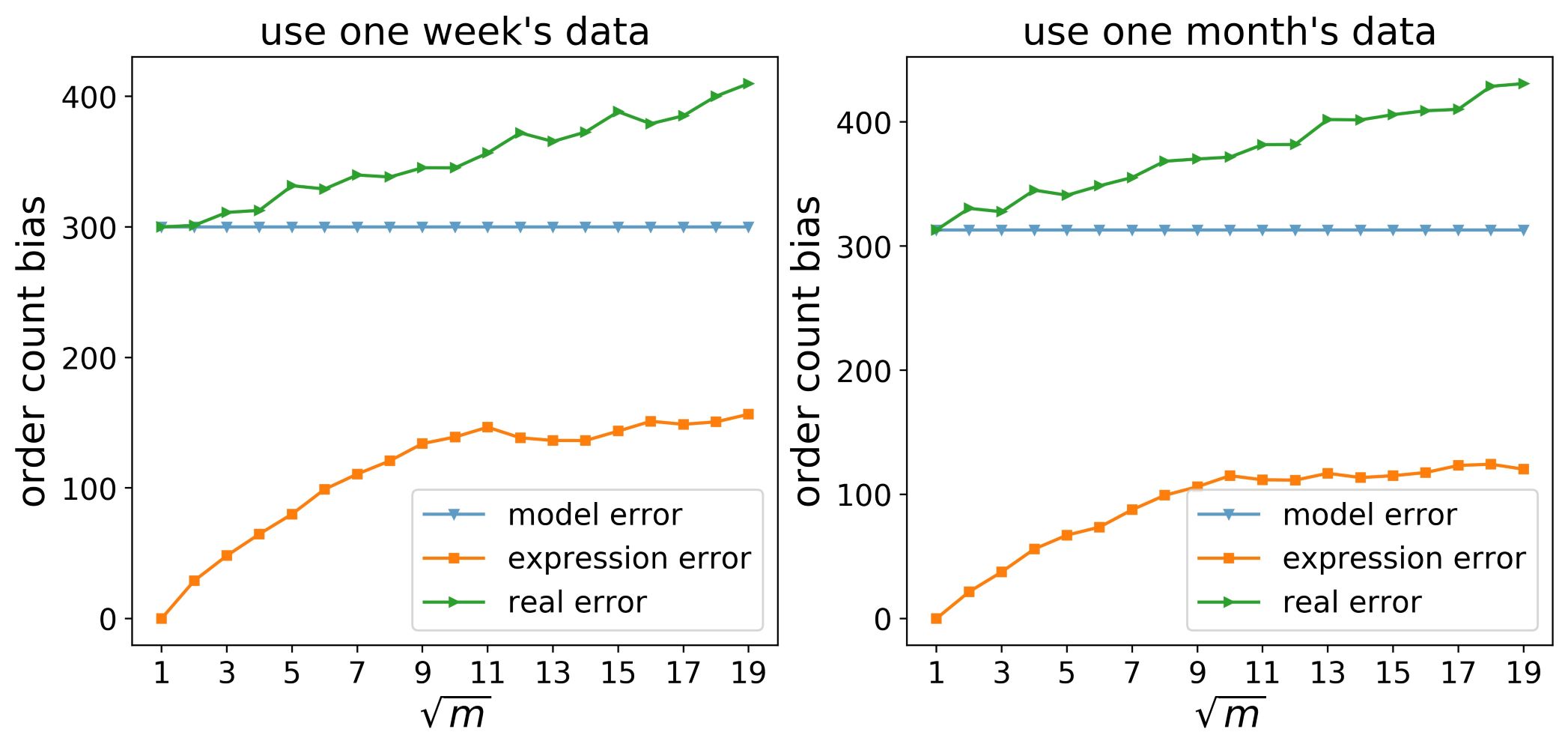}}
	\caption{\small Effect of $m$ on Expression Error, Model Error, Real Error}
	\label{fig:m_error}
\end{figure}
\subsection{Results on Calculation of Expression Error}
We  explore the performance of Algorithms \ref{algo:repre_algorithm} and \ref{algo:simple_repre_algorithm} on the calculation of expression error. We $N=128\times 128$, $n=16\times 16$ and $m=8\times8$. In this case, we calculate the expression error of a HGrid by Algorithms \ref{algo:repre_algorithm} and \ref{algo:simple_repre_algorithm}. Theorem \ref{the:arbitrary_precision} proves that the accuracy of expression error calculation increases with the increase of $K$. However, the increase in $K$ is accompanied by a rapid increase in computing costs. Figure \ref{fig:cal_for_representation} shows that the cost of the most straightforward algorithm (i.e., no optimizations are made) increases rapidly with the increase of $K$, and the calculation cost of Algorithm \ref{algo:repre_algorithm} is linearly related to $K$. However, the calculation time of Algorithm \ref{algo:simple_repre_algorithm} is always kept at a low level, which indicates the effectiveness of the algorithm.
\begin{figure}
	\scalebox{0.40}[0.40]{\includegraphics{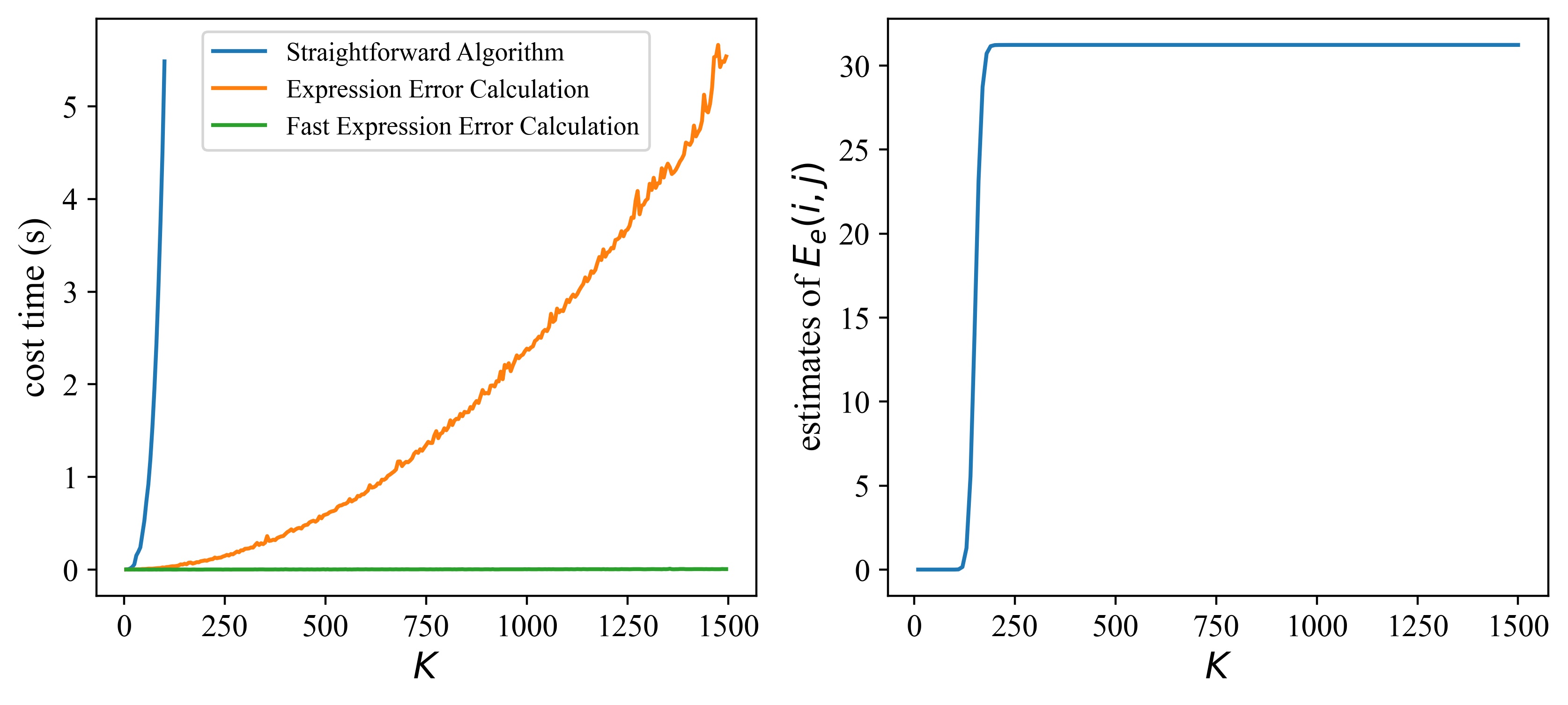}}
	\caption{\small Effect of $K$ on efficiency for computing expression and accuracy of Algorithm \ref{algo:simple_repre_algorithm}}
	\label{fig:cal_for_representation}
\end{figure}

We generally choose to set $K$ as 250 based on the result of Figure \ref{fig:cal_for_representation} to obtain a more accurate result of expression error while avoiding too much computing cost.
\subsection{Influence of $bound$ on the Performance of Iterative Method}

\begin{figure}[t!]
	\begin{minipage}{0.46\linewidth}\centering
		\scalebox{0.2}[0.2]{\includegraphics{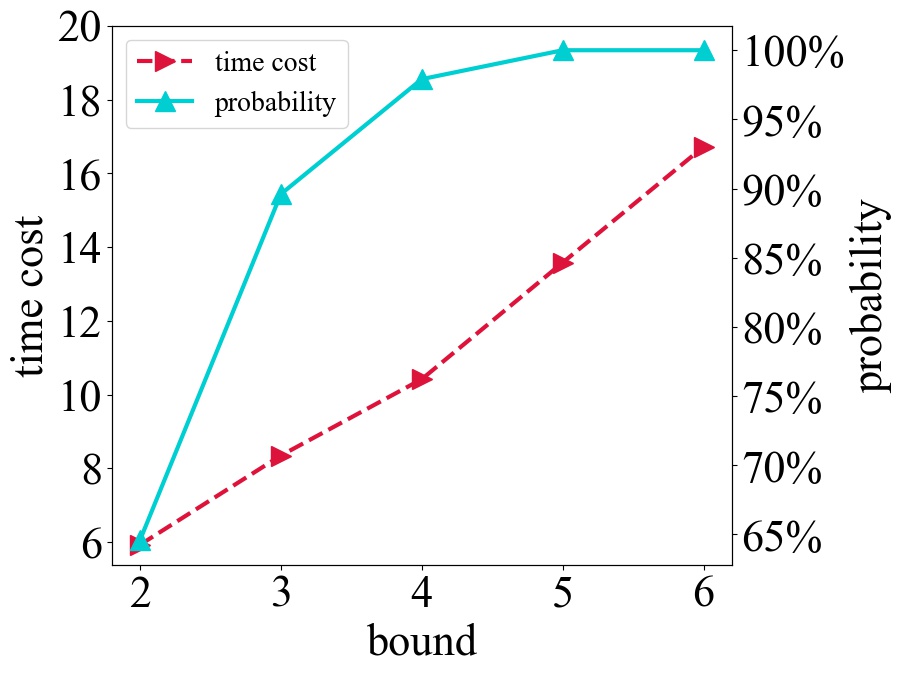}}
		\caption{\small Effect of $bound$ on the Solution and the Cost of Algorithm \ref{algo:opg_3}}
		\label{fig:bound}
	\end{minipage}
\hspace{3ex}
    \begin{minipage}{0.46\linewidth}\centering
		\scalebox{0.22}[0.22]{\includegraphics{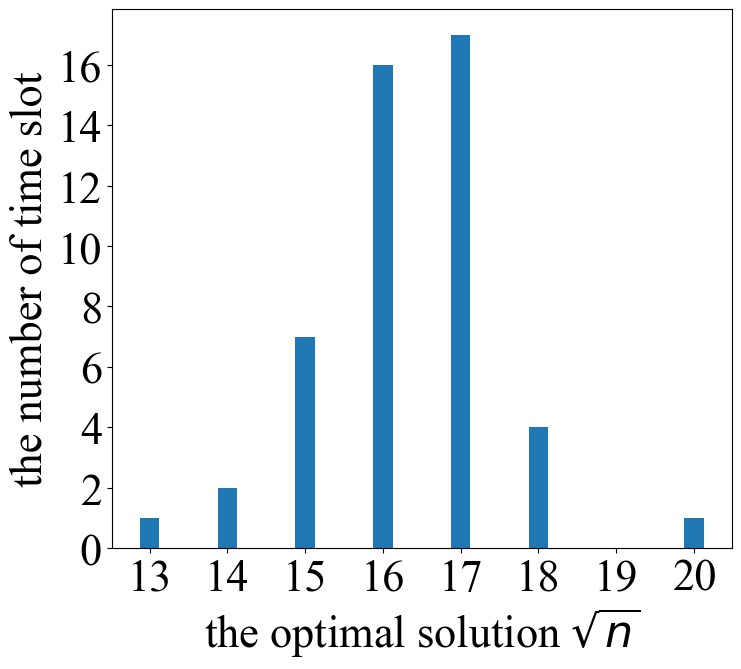}}
		\caption{\small Distribution of Optimal solutions in different time slots}
	    \label{fig:solution}
    \end{minipage}
\end{figure}

\begin{figure}[t!]
	\centering
	\scalebox{0.7}[0.7]{\includegraphics{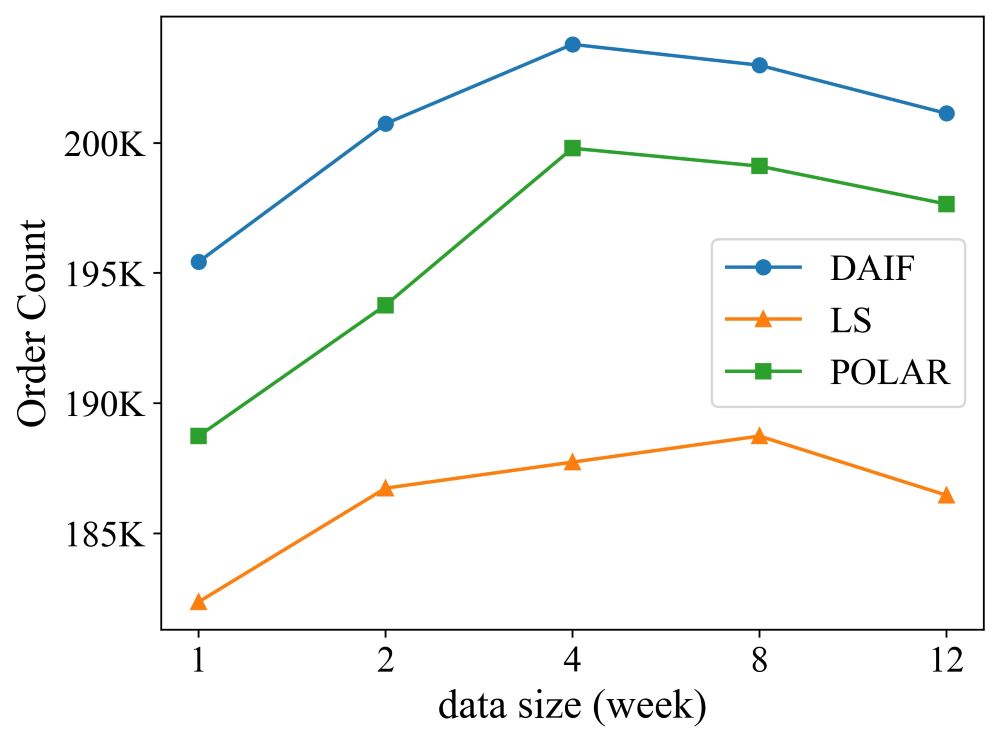}}
	\caption{\small Effect of the Size of Dataset on the Performance of Different Crowdsourcing Algorithm}
	\label{fig:vary_datasize}
\end{figure}

Figure \ref{fig:bound} shows the effect of the selection of bound on Algorithm \ref{algo:opg_3}. With the increase of $bound$, the probability of Algorithm \ref{algo:opg_3} finding the optimal solution increases gradually, and the cost of the algorithm also increases. In addition, Figure \ref{fig:solution} shows the distribution of optimal selections of $n$ in 48 periods of a day, which shows that the optimal value of $n$ is 17 for most times.

\subsection{Effect of the size of dataset}
The size of dataset should not be too large or too small. For example, 3 months data for training can harm the performance as the distribution may change, thus the estimation of $\alpha_{ij}$ can be inaccurate and cannot find the optimal $n$. As for training with too small dataset (e.g., one week), the performance also drops since the data is not sufficient to train the prediction models. As shown in Figure \ref{fig:vary_datasize}, the results are best when we use 4 weeks' data as the training set.

\end{document}